\documentclass{LMCS}
\usepackage{proof}
\usepackage{ifthen}
\usepackage{amsmath, amssymb}
\usepackage{graphicx,enumerate,hyperref}

\theoremstyle{plain}
\newtheorem{lemma}[thm]{Lemma}

\newcommand{\judgement}[3]{\ensuremath{#1\ \vdash_{\Sigma} #2\!:\! #3}}
\newcommand{\judgements}[4]{#1\ \vdash_{#4} #2:\, #3}
\newcommand{\tf}[3]{#1:\, #2\,\rightarrow\, #3}
\newcommand{\st}[2]{\mathsf{L}_{#2}(#1)}
\newcommand{\tv}[2]{#1\!:\!\:\!#2}

\newcommand{\restrict}[2]{#1|_{#2}}

\newcommand{\hdarg}{\mathit{hd}}
\newcommand{\tlarg}{\mathit{tl}}
\newcommand{\xarg}{\mathit{x}}
\newcommand{\yarg}{\mathit{y}}

\newcommand{\cop}{\ensuremath{\mathsf{copy}}}

\newcommand{\append}{\ensuremath{\mathsf{append}}}

\newcommand{\prgs}{\ensuremath{\mathsf{progression}}}

\newcommand{\pairs}{\ensuremath{\mathsf{pairs}}}

\newcommand{\cprod}{\ensuremath{\mathsf{cprod}}}
\newcommand{\sqdiff}{\ensuremath{\mathsf{sqdiff}}}

\newcommand{\ttt}{\ensuremath{\mathsf{True}}}

\newcommand{\f}{\ensuremath{\mathsf{f}}}

\newcommand{\sizeexprs}{\mathit{SizeExpr}}
\newcommand{\exprs}{\mathit{Expr}}
\newcommand{\basics}{\mathit{Basic}}
\newcommand{\types}{\mathit{Types}}
\newcommand{\ftypes}{\mathit{FTypes}}
\newcommand{\gtypes}{\mathit{GTypes}}
\newcommand{\typevars}{\mathit{TypeVar}}
\newcommand{\sizevars}{\mathit{SizeVar}}
\newcommand{\expvars}{\mathit{ExpVar}}

\newcommand{\func}[3]{\begin{array}{@{}l@{}}\mathsf{#1}\ #2 =\\ \quad \begin{array}[t]{@{}l@{}}#3\end{array}\end{array}}

\newcommand{\lettt}[3]{\begin{array}[t]{@{}l@{}}\mathsf{let}\ #1 = #2\\ \mathsf{in}\ \begin{array}[t]{@{}l@{}}#3\end{array}\end{array}}
\newcommand{\lete}[3]{\mathsf{let}\ #1 = #2 \ \mathsf{in}\ #3}
\newcommand{\letfune}[4]{\mathsf{letfun}\ #1(#2) = #3\ \mathsf{in}\ #4}
\newcommand{\letexterne}[3]{\mathsf{letextern}\ #1(#2) \ \mathsf{in}\ #3}
\newcommand{\letfun}[4]{\mathsf{letfun}\ #1(#2) = #3\\ \mathsf{in}\ #4}

\newcommand{\fapp}[2]{\mathsf{#1}(#2)}
\newcommand{\cons}[2]{\mathsf{cons}(#1, #2)}
\newcommand{\nil}{\mathsf{nil}}

\newcommand{\ifte}[3]{\mathsf{if}\ {#1}\ \mathsf{then}\  {#2}\  \mathsf{else}\  {#3}}

\newcommand{\matchshort}[5]{
\begin{array}[t]{@{}l@{\ }l@{\ }l@{}}
\mathsf{match}\ #1\ \mathsf{with} & \shortmid  & \nil \Rightarrow \begin{array}[t]{@{}l@{}l@{}}#4\end{array}\\
 &\shortmid & \cons{#2}{#3} \Rightarrow \begin{array}[t]{@{}l@{}l@{}}#5\end{array}\\
\end{array}}

\newcommand{\match}[5]{
\begin{array}[t]{@{}l@{\ }l@{\ }l@{}}
\mathsf{match}\ #1\ \mathsf{with} & |  & \nil \Rightarrow \begin{array}[t]{@{}l@{}l@{}}#4\end{array}\\
 &|& \cons{#2}{#3} \Rightarrow \begin{array}[t]{@{}l@{}l@{}}#5\end{array}\\
\end{array}}

\newcommand{\matchline}[5]{
\mathsf{match}\ #1\ \mathsf{with}\ |\ \nil \Rightarrow #4 \ |\
\cons{#2}{#3} \Rightarrow #5}

\newcommand{\typesvg}{\ensuremath{\tau^{\bullet}}}
\newcommand{\typesv}{\ensuremath{\tau}}
\newcommand{\typesvs}{\ensuremath{\tau^{\circ}}}
\newcommand{\typesf}{\ensuremath{\tau^{f}}}
\newcommand{\sz}{\ensuremath{\mathit{p}}}
\newcommand{\context}{\ensuremath{\Gamma{}}}
\newcommand{\sign}{\ensuremath{\Sigma{}}}
\newcommand{\intt}{\ensuremath{\mathtt{Int}}}
\newcommand{\FVS}[1]{\ensuremath{\mathit{FVS}({#1})}}

\newcommand{\closures}{\ensuremath{\mathcal{C}}}
\newcommand{\closExt}{\ensuremath{\mathcal{E}}}
\newcommand{\opsem}[5]{\ensuremath{{#1};\, {#2}\ \vdash \ {#3} \ \leadsto\, {#4};\, {#5}}}

\newcommand{\dom}[1]{\mathit{dom}(#1)}

\newcommand{\Val}{\mathit{Val}}
\newcommand{\Loc}{\mathit{Loc}}
\newcommand{\St}{\mathit{Store}}
\newcommand{\Hp}{\mathit{Heap}}

\newcommand{\nul}{\mathtt{NULL}}
\newcommand{\hd}{\mathtt{hd}}
\newcommand{\tl}{\mathtt{tl}}

\newcommand{\mt}[3]{\ensuremath{{#1}:\,{#2}  \longrightarrow {#3}}}

\newcommand{\R}{\ensuremath{\mathcal{R}}}
\newcommand{\Z}{\ensuremath{\mathcal{Z}}}
\newcommand{\Ps}{\ensuremath{\mathcal{P}}}
\newcommand{\Reg}[2]{\ensuremath{\R\ifthenelse{\equal{#1}{}}{}{({#1},\ {#2})}}}
\newcommand{\FV}[1]{\ensuremath{\mathit{FV}({#1})}}

\newcommand{\mdl}[4]{\ensuremath{{#1}\ \models^{#2}_{#3}\ {#4}}}

\newcommand{\rstrloc}[2]{\ensuremath{{#1}|_{\dom{#1}\setminus \{#2\}}}}

\newcommand{\soundval}[3]{\ensuremath{\mathit{Valid}_\mathsf{val}\ifthenelse{\equal{#1}{}}{}{(#1, #2, #3)}}}
\newcommand{\soundstore}[4]{\ensuremath{\mathit{Valid}_\mathsf{store}\ifthenelse{\equal{#1}{}}{}{(#1, #2, #3, #4)}}}
\newcommand{\val}{\ensuremath{\epsilon}}

\def\doi{5 (2:10) 2009}
\lmcsheading%
{\doi}
{1--35}
{}
{}
{Jan.~\phantom{0}4, 2006}
{May~\phantom{.}25, 2009}
{}

\begin{document}

\title[Polynomial Size Analysis of First-Order Shapely
  Functions]{Polynomial Size Analysis \\ of First-Order Shapely
  Functions\rsuper *}

\author[O.~Shkaravska]{Olha Shkaravska\rsuper a}
\address{{\lsuper a}Institute for Computing and Information Sciences, Radboud University Nijmegen}
\email{shkarav@cs.ru.nl}
\thanks{{\lsuper a}This research is sponsored by the Netherlands Organisation for Scientific
Research (NWO), project Amortised Heap Space Usage Analysis (AHA),
grantnr. 612.063.511.}

\author[M.~v.~Eekelen]{Marko van Eekelen\rsuper b}
\address{{\lsuper b}Institute for Computing and Information Sciences, Radboud University Nijmegen \emph{and}
Faculty of Information Science, Open University of the Netherlands}
\email{marko@cs.ru.nl \emph{and} marko.vaneekelen@ou.nl}

\author[R.~v.~Kesteren]{Ron van Kesteren\rsuper c}
\address{{\lsuper c}Alten Nederland, Consulting and Engineering in Advanced Technology}
\email{ronvankesteren@gmail.com}

\titlecomment{{\lsuper*}This paper is an extended version of the paper \cite{ShvKvE07b}, presented at the TLCA conference in 2007. The paper is extended by the full soundness proof and the full
presentation of the type-inference procedure from \cite{vKShvE07}. Moreover, in the presented version we consider rational size polynomials instead of integer ones.}


\begin{abstract}
We present a size-aware type system for first-order shapely function
definitions. Here, a function definition is called \emph{shapely} when the
size of the result is determined exactly by a polynomial in the sizes of
the arguments. Examples of shapely function definitions may be
implementations of matrix
multiplication and the Cartesian product of two lists.

The type system
is proved to be sound w.r.t. the operational semantics of the language.
The type checking problem is shown to be undecidable in general.
We define a natural syntactic restriction such that the
type checking becomes decidable, even though size polynomials are not
necessarily linear or monotonic.

Furthermore, we have shown that the type-inference problem is at least
semi-decidable (under this restriction). We have implemented a procedure
that combines run-time testing and type-checking to automatically obtain
size dependencies. It terminates on total typable function definitions.
\end{abstract}

\maketitle
\vfill

\section{Introduction}
\label{Introduction}

We explore typing support for checking size dependencies for
\textit{shapely} first-order function definitions (functions for short).
The shapeliness of these
functions lies in the fact that the size of the result is a
polynomial in terms of the arguments' sizes.

\subsection{Variety of resource analysis techniques}

This research is a part of the Amortised Heap Space Usage Analysis (AHA) project
\cite{vEShvK07}.  Estimating heap consumption is an active research area as it becomes more and more
of an issue in many applications, including programming for small devices, e.g. smart cards, mobile phones,
embedded systems and distributed computing.

Amortization is a promising technique to obtain accurate bounds of resource
consumption and gain. An amortised estimate of a resource does not target a single
operation but a sequence of operations. One assigns some amortised cost to an operation.
This amortised cost may be higher or lower than the operation's actual cost.
For the sequence considered, it is important that
its overall amortised cost covers its overall actual cost.
An amortised cost of the sequence lies between its actual cost and the simple multiplication of the worst-case of one operation by the length of the sequence.
An amortised cost of the sequence is in many cases easier to compute
than its actual cost and it is obviously better than the worst-case estimate.

Combining amortization with type theory allows to infer linear
heap-consumption bounds for functional programs with explicit memory
deallocation \cite{POPL03}. The \textit{AHA} project aims to adapt this method for
\emph{non-linear} bounds within (lazy) functional programs
and transfer the results to the object-oriented programming. Contrary
to linear amortised bounds, to obtain non-linear heap estimates
one does need to know sizes of structures that takes part in computation,
see, for instance \cite{vEShvK07}.

The AHA project seems to be part of an emerging trend since a growing number of works
are addressing resource analysis. Here we mention some of them.

In \cite{AmZil06} the authors develop new method to statically (polynomially) bound the resources needed
for the execution of systems of concurrent threads. The method generalises an approach designed
for first-order functional languages that relies on a combination of standard termination techniques
for term rewriting systems and an analysis of the size of the computed values based on the notion
of a polynomial quasi-interpretation. Quasi-interpretations were applied to size analysis firstly in \cite{BMM05b}.
In \cite{AvMoSch08} the authors describe a fully automated tool that implements a few techniques
that directly classify run-time complexity (i.e. techniques that use the number of
rewrite steps as complexity measure), including polynomial quasi-interpretations.

Several groups have studied programming languages with \textit{implicit computational complexity}
(ICC) properties. This line of research is motivated both by the perspective of
automated complexity analysis, and by foundational goals,
in particular to give natural characterisations of complexity classes, like PTIME or PSPACE.
In \cite{G92} characterisation of PTIME is given in terms
of bounded linear logic.
In  \cite{GabMarRon08} one proposes a characterization of PSPACE by means of an extension of
(soft affine) typed lambda calculus. For this extension, the authors design a call-by-name
evaluation machine in order to compute programs in polynomial space.
In \cite{AtBailTer07} one addresses the problem of
typing lambda-terms in a variant of second-order light linear logic. The authors give a procedure which, starting
with a term typed in system F, determines whether it is typable in the logic. It is
shown that the procedure can be run in time polynomial in the size of the original Church typed system F term.

Resource analysis may be performed within a \textit{Proof Carrying Code} framework.
In \cite{AsMcK06} one introduces the notion of a resource policy for mobile
code to be run on smart devices.
Such a resource policy is integrated in a  proof-carrying code architecture.
Two forms of policy are used: guaranteed policies which come with proofs
and target policies which describe limits of the device.

In \cite{AlArGenPuebZan07} one describes resource consumption
for Java bytecode by means of
Cost Equation Systems (CESs), which are similar to, but
more general than recurrence equations. CESs express the
cost of a program in terms of the size of its input data.
In a further step, a closed form (i.e., non-recursive) solution or upper
bound can sometimes be found by using existing Computer Algebra Systems, such as Maple
and Mathematica.
This work is continued by the authors in \cite{AlArGenPueb08},
where mechanisms of
constructing solutions  of CESs and upper bounds are studied closely.
They consider monotonic cost expressions only.

In \cite{Ben01} the author describes the Automated Complexity Analysis Prototype (ACAp) system  for automated time analysis of functional programs. Symbolic evaluation of recursive programs generates systems of multi-variable difference equations, which are solved using Mathematica.

In \cite{GuMeCh09} the authors describe a technique for computing symbolic bounds on the number of statements
a procedure executes in terms of its inputs and user defined size functions. The technique is based
on multiple counter instrumentation that allows to compute linear bounds individually for each counter.
The bounds on these counters are then composed to generate total bounds that are non-linear and disjunctive.

\subsection{Exploring size dependencies}

In this paper we restrict our attention to a language
with polymorphic lists as the only data-type. For such a language,
this paper develops a size-aware type system for which we define a
fully automatic type checking and inference procedure.

A typical example of a shapely function in this language is
$\cprod$ that computes the Cartesian product of two sets,
stored as lists. It is given below.
The auxiliary function $\mathsf{pairs}$ creates pairs of a single
value and the elements of a list. To get a Cartesian product the function
$\mathsf{cprod}$ does this for all elements from the first list
separately and appends the resulting intermediate lists. Furthermore, the
function definition of $\mathsf{append}$ is assumed:
$$
\begin{array}{lll}
\cprod(l_1, l_2) & = &
\match{l_1}{\hdarg}{\tlarg}{\nil} {\append(\pairs(\hdarg, l_2),\
\cprod(\tlarg, l_2))}
\end{array}
$$
\noindent where

$$
\begin{array}{l}
\pairs(\xarg, l) =
\match{l}{\hdarg}{\tlarg}{\nil}{
\lettt{l'}{\cons{\xarg}{\cons{\hdarg}{\nil}}}
{\cons{l'}{\pairs(\xarg,\,\tlarg)}}}
\end{array}
$$

Given two lists, for
instance $[1,\ 2,\ 3]$ and $[4,\ 5]$, it returns the list with all
pairs created by taking one element from the first list and one
element from the second list: $[[1,\,4],\ [1,\,5],\ [2,\,4],\
[2,\,5],\ [3,\,4],\ [3,\,5]]$. Hence, given two lists of length $n$
and $m$, it always returns a list of length $n  m$ containing
pairs. This is expressed by the type $\st{\alpha}{n}\times
\st{\alpha}{m} \,\rightarrow\, \st{\st{\alpha}{2}}{n * m}$.

Shapeliness is restrictive, but it is an important foundational step.
It makes type checking decidable in the non-linear case and it allows to
infer types ``out-of-the-box'', since experimental points are positioned
exactly on the graph of the polynomial. Exact sizes will be used in
future work to derive lower/upper bounds on the output sizes. We need such bounds for
investigating amortised resource bounds in the AHA project.
Nonlinear amortised resource consumption
relies on the size of input data, and its gain is calculated based on the size of output.

In this paper our only concern is in sizes of input and output. For instance, the time and space complexity of a function definition with a polynomial input-output size dependency may
exceed polynomial space and time consumption due to internal structures and computations.

\subsection{Related work on size analysis}

Information about input-output size dependencies is applied to
time and space analysis and optimization, because run time and
heap-space consumption obviously depend on the sizes of the data
structures involved in the computations. Knowledge of the exact size
of data structures can be used to improve heap space analysis for
expressions with destructive pattern matching. Amortised heap space
analysis has been developed for linear bounds by Hofmann and Jost
\cite{POPL03}. Precise knowledge of sizes is required to extend this
approach to non-linear bounds. Another application of
exact size information is \textit{load distribution} for parallel
computation. For instance, size information helps to distribute a
storage effectively and to safely store vector fragments
\cite{Chat90}.

The analysis of (exact) input-output size
dependencies of functions itself has been explored in a series of works.
Some interesting work on shape analysis has been done by Jay and
Sekanina~\cite{Jay97}. In this work, a shapely program expression is translated
into a corresponding abstract program expression over sizes. Thus, the
dependency of the result size on the argument sizes has the form of
a program expression. However, deriving an arithmetic function
from it is beyond the scope of their work.

Functional dependencies of sizes in a \textit{recurrent form} may be
derived via program analysis and transformation, as in the work of
Herrmann and Lengauer \cite{Her01}, or through a type inference
procedure, as presented by Vasconcelos and Hammond~\cite{Vas03}.
Both results can be applied to non-shapely functions, higher-order
functions and non-linear size expressions. However, solving the
recurrence equations to obtain a closed-form solution is left as an
open problem for external solvers.
In the second paper monotonic bounds are studied.

To our knowledge, the only work yielding closed-form solutions for
size dependencies is limited to monotonic dependencies. For instance,
in the well-known work of Pareto \cite{Par98}, where
\textit{non-strict} sized types are used to prove termination,
monotonic linear upper bounds are inferred.
There linearity is a sufficient condition for the type checking procedure
to be decidable. In the series of works on polynomial quasi- \cite{BMM05b}
and sup-interpretations \cite{MP06}
one studies max-polynomial upper bounds. The checking and inference rely on
real arithmetic. In general, (inference) synthesis procedures are exponential
w.r.t. the size of a program.
For \textit{multilinear} polynomials in \textit{max}-\textit{plus}-algebra it is shown to be
of polynomial complexity \cite{Am05}.

Our approach differs two-fold.
Firstly, quasi-interpretations give monotonic bounds. With non-monotonic size dependencies
polynomial quasi-interpretations may lead to significant over-estimations.
Secondly, to get exact bounds we use rational arithmetic instead of
real arithmetic. Our motivation for this choice lies in the fact that one should use decidability procedures in reals
with care, if one applies them to integers or naturals. For instance,
$x^2 \leq x^3$ holds in naturals, but  not in reals, since it does not hold on
$0 < x < 1$.

The approaches summarized in the previous paragraphs either leave
the (possibly undecidable) solving of recurrences as a problem
external to their approach, or are limited to
monotonic dependencies.

\subsection{Content of the paper}

In this work, we go beyond monotonicity and linearity and consider a type checking
procedure for a first-order functional programming language (section
\ref{language}) with polynomial size dependencies (section
\ref{typing}).

In subsection \ref{zero-order types} we define
zero-order types and their set-theoretic semantics.
In subsections \ref{first-order types} and \ref{typing rules}
we define first-order types  and give typing rules respectively.
The soundness of type system w.r.t. the operational semantics of the language
is studied  in subsection \ref{soundness}. The type system is not complete in the class
of all shapely functions, and no such complete system exists (subsection \ref{completeness}).

In section \ref{typechecking} we show that type checking is reduced to the
entailment checking over Diophantine equations. Type checking is
shown to be undecidable in general (subsection \ref{typechecking:undec}).
However, type-checking is decidable under certain syntactic condition
for function bodies (subsection \ref{typechecking:dec}).

We define in detail a method for type inference in section \ref{inference}.
It terminates on a nontrivial class of shapely functions.
It does not terminate when
either the function under consideration does not terminate,
or it is not shapely, or its correct size dependency is rejected
by the type-checker due type-system's incompleteness.

Finally, in section \ref{conclusion} we
overview the results and discuss further work.


\section{Language}
\label{language}

The typing system is designed for a first-order functional language
over integers and (polymorphic) lists.

The syntax of language expressions is defined by the following
grammar (the example in the introduction used a sugared
version of this syntax):
$$
\begin{array}{llll}
\basics & b & ::= &  c \ |\  x\,\mathsf{binop}\,y\
| \ \nil \ | \ \cons{z}{l} \ |\ f(z_1, \ldots, z_n)\\
\exprs & e & ::= &  \ b\\
& & &  |\ \lete{z}{b}{e_1}\\
& & &  |\ \ifte{\xarg}{e_1}{e_2}\\
& &  & | \ \matchshort{l}{z}{l'}{e_1}{e_2} \\
& & & | \ \letfune{f}{z_1, \ldots, z_n}{e_1}{e_2}\\
& &  & | \ \letexterne{f}{z_1, \ldots, z_n}{e_1} \\
\end{array}
$$
\noindent
where $c$ ranges over integer constants,
$z$, $x$, $y$, $l$ denote zero-order program variables ($x$ and $y$ range over integer
variables, $l$ possibly decorated with sub- ans superscripts,
ranges over lists and $z$ ranges over
program variables when their types are not relevant),
$\mathsf{binop}$ is one of the four integer binary operations:
$+,\, -,\, \mathsf{div},\,\mathsf{mod}$,
and $f$ denotes a function name.

The syntax distinguishes between zero-order let-binding of variables
and first-order letfun-binding of functions. In a function body, the
only free program variables that may occur are its parameters:
$\FV{e_1} \subseteq \{ z_1, \ldots, z_n \}$. The operational
semantics is standard, therefore the definition is postponed until
it is used to prove soundness (section \ref{soundness}).

We  prohibit head-nested let-expressions and restrict sub-expressions in
function calls to variables to make type-checking straightforward.
Program expressions of a general form may be equivalently transformed to
expressions of this form. It is useful to think of the presented language
as an intermediate language.

For practical reasons and in order to support
modularity, we introduce a $\mathsf{letextern}$ \textit{declaration}, which
makes it possible to
call functions implemented in other modules that may be defined in other languages.


\section{Type System}
\label{typing}

We consider a type system, constituted from  zero- and first-order types,
corresponding typing rules for program constructs and Peano arithmetic extended
to rational numbers as (classes of equivalence of) pairs of integers, rational addition and multiplication\footnote{
Rational addition is defined as
$\dfrac{a}{b}+\dfrac{c}{d}=\dfrac{ad+cb}{bd}$. Rationals with their addition and multiplication form
a field, more precisely a field of integer fractions.}.

\subsection{Zero-order types and their semantics}
\label{zero-order types}

Sized types are derived using a type and effect system in which
types are annotated with size expressions. Size expressions are
polynomials representing lengths of finite lists and arithmetic
operations over these lengths:
\vspace{0.1cm}

$
\begin{array}{llll}
 \sizeexprs & \sz & ::= & \mathcal{Q}\ | \ n \ | \  \sz
+ \sz \ |\ \sz - \sz \ |\ \sz
* \sz
\end{array}
$
\vspace{0.1cm}

\noindent
where $\mathcal{Q}$ denotes rational numbers, and $n$, possibly decorated with sub- and superscripts,
denotes a size variable,
which stands for any concrete size (natural number). For any natural
number $k$, $n^k$ denotes the $k$-fold product $n * \ldots * n$.

Size expressions are rational polynomials that map natural numbers into
natural numbers. For instance, the polynomial
$p(n)=\dfrac{n(n+1)}{2}$ represents the size dependency of the function
$\prgs$:
$$
\begin{array}{lll}
\prgs(l) & = &
\match{l}{\hdarg}{\tlarg}{\nil} {\append(\prgs(\tlarg),\,l)}
\end{array}
$$
\noindent For example, it maps $[1,\,2,\,3]$ on $[3,\,2,\,3,\, 1,\,2,\,3]$. The output size dependency
is given by the arithmetic progression $0+1 + \ \ldots \ +(n-1) + n$, where
$n$ is the size of an input. This explains the name of the function \cite{vKShvE07}.

Zero-order types are assigned to program values, which are interpreted
as integer numbers and finite lists. A list type is annotated with a size expression
that represents the length of the list:
\vspace{0.1cm}

$
\begin{array}{llll}
\types & \typesv & ::= & \intt\ | \ \alpha\ | \ \st{\typesv}{\sz}
\end{array}$
\vspace{0.1cm}

\noindent
where $\alpha$
is a type variable. This structure entails that if the elements
of a list are lists themselves, then all these element-lists must be
of the same size. Thus, instead of lists it would be more precise to
talk about matrix-like structures. For instance, the type
$\st{\st{\intt}{2}}{6}$ is given to a list whose elements are all
lists of exactly two integers, such as $[[1,4], [1,5], [2,4], [2,5],
[3, 4], [3, 5]]$.

It is easy to see  that for all $m$ the types
$\st{\st{\intt}{m}}{0}$ are equal, because they represent the singleton containing
$[\,]$. The same holds for
$\st{\st{\alpha}{m}}{0}$.  This induces a natural equivalence relation
on types. For instance
$\st{\st{\st{\alpha}{p}}{0}}{q} \equiv \st{\st{\st{\alpha}{p'}}{0}}{q}$.
The equivalence expresses the fact
that the size of a list is not relevant when such a list does not exist, because an outer list is
empty. Now, we define formally
an entailment $D \vdash \tau = \tau'$, where $D$ is a conjunction
of equations between polynomials. The definition is inductive on $\tau$.
The entailment $D \vdash \tau = \tau'$ holds if and only if
\begin{enumerate}[$\bullet$]
\item $\tau = \tau' = \intt$ or $\tau = \tau' = \alpha$ for some type variable
$\alpha$;
\item $\tau = \st{\tau''}{p}$ and $\tau'=\st{\tau'''}{p'}$ have the same underlying type (i.e.
the type with  annotations omitted) and
\begin{enumerate}[(1)]
\item $D \vdash p=p'$, and
\item $D \vdash p= 0$ or $D \vdash \tau''= \tau''', $
\vspace{0.2em}
\end{enumerate}
\end{enumerate}
\noindent with $D \vdash p=q$ being an arithmetical entailment, meaning
$\forall \ \bar{n}. D(\bar{n}) \rightarrow  p(\bar{n}) =q(\bar{n})$, where $\bar{n}$
is the collection of all size variables taken from $D$, $q$ and $p$. For instance,

$$
\begin{array}{l}
m=0 \vdash \st{\alpha}{n+m} =\st{\alpha}{n}\quad\textrm{and}\\
m-1=0,\,n=0 \vdash \st{\st{\alpha}{2}}{n+m-1} =\st{\st{\alpha}{3}}{n}\\
\end{array}
$$
\noindent
hold, whereas $n=0 \vdash \st{\st{\alpha}{2}}{n+m-1} =\st{\st{\alpha}{3}}{m-1}$ does not.

The sets $\FV{\typesv}$ and $\FVS{\typesv}$ of the
free type and size variables of a type $\typesv$
are defined inductively in the obvious way. Note, that
$\FVS{\st{\st{\alpha}{m}}{0}}=\emptyset$, since
the type is equivalent to $\st{\st{\alpha}{0}}{0}$.

Zero-order types without size or type variables are ground types:
\vspace{0.1cm}

$\begin{array}{llll} \gtypes & \typesvg & ::= & \typesv \mbox{\
such that\ }\FVS{\typesv} = \emptyset \wedge \FV{\typesv} =
\emptyset
\end{array}$
\vspace{0.1cm}

In our semantic model a heap is essentially a collection of
locations $\ell$ that can store list elements. A location is the
address of a cons-cell each consisting of a $\hd$-field, which
stores the value of a list element, and a $\tl$-field, which
contains the location of the next cons-cell of the list (or the
$\nul$ address). Formally, a program value is either an integer
constant, a location, or the $\nul$-address.
A heap is a finite partial mapping from locations and fields to
program values:

$$
\begin{array}{l}
\Val\ v \ ::= \  c\ |\ \ell\ |\ \nul \quad\quad\quad \ell \in \Loc \quad\quad c \in \intt\\
\textit{Hp} \ h \  : \  \Loc \rightharpoonup \left\{ \hd, \tl
\right\} \rightharpoonup \Val\\
\end{array}
$$

We will
write $h.\ell.\hd$ and $h.\ell.\tl$ for the results of applications
$h \ \ell\ \hd$ and $h\  \ell \ \tl$, which denote
the values stored in the heap $h$ at the location $\ell$ at fields $\hd$ and $\tl$,
respectively. Let $h[\ell.\hd:=v_h,\
\ell.\tl:=v_t]$ denote the heap equal to $h$ everywhere but in $\ell$,
which at the $\hd$-field of $\ell$ gets value $v_h$ and at the
$\tl$-field of $\ell$ gets value $v_t$.

The semantics $w$ of a program value $v$ is a set-theoretic
interpretation with respect to a specific heap $h$ and a ground
type $\tau$. It is given via the four-place relation $\mdl{v}{h}{\typesv}{w}$,
where integer constants interprets themselves, and locations are
interpreted as non-cyclic lists:

$$
\begin{array}{l@{\ }l@{\ }l}
c & \models^h_\intt & c\\
\nul & \models^h_{\st{\typesvg}{0}} & \texttt{[]} \\
\ell & \models^{h}_{\st{\typesvg}{n^\bullet}} & w_\hd :: w_\tl \ \textsf{iff}
\
\begin{array}[t]{l}
                                                     n \geq  1, \ell \in\dom{h},\\
                                                    \mdl{h.\ell.\hd}{\rstrloc{h}{\ell}}{\typesvg}{w_\hd},\\
                                                    \mdl{h.\ell.\tl}{\rstrloc{h}{\ell}}{\st{\typesvg}{n^\bullet-1}}{w_\tl}
                                                    \end{array}\\
\end{array}
$$

\noindent where $n^\bullet$ is a natural constant and $\rstrloc{h}{\ell}$ denotes the heap equal to $h$
everywhere except for $\ell$, where it is undefined.

\subsection{First-order types}
\label{first-order types}

First-order types are assigned to shapely functions over values of a
zero-order type. Let \typesvs{} denote a zero-order type of which
the annotations are all size variables. First-order types are then
defined by:

$$\begin{array}{llll}
\ftypes & \typesf & ::= & \typesvs_1 \times \ldots \times \typesvs_n
\rightarrow \typesv_{n+1} \\ &&& \mbox{such\ that\ }
\FVS{\typesv_{n+1}} \subseteq \FVS{\typesvs_1} \cup \cdots \cup
\FVS{\typesvs_n}\\
\end{array}$$

For instance, one expects
that the following function definitions
(in the sugared syntax\footnote{In the sugared syntax we
use $f(g(z))$ for ``$\lete{z'}{g(z)}{f(z')}$''})
will be well-typed in
the system:

\vspace{1.0em}

\noindent
$
\begin{array}{l}
\tf{\append}{\st{\alpha}{n}\times \st{\alpha}{m}}{\st{\alpha}{n+m}}\\
\append(l_1,\,l_2)=
\match{l_1}{\hdarg}{\tlarg}{l_2}{\cons{\hdarg}{\append(\tlarg,\,l_2)}} \\
\end{array}
$

\vspace{1.0em}

\noindent
$
\begin{array}{l}
\tf{\pairs}{\alpha \times \st{\alpha}{n}}{\st{\st{\alpha}{2}}{n}}\\
\pairs(\xarg, l) =
\match{l}{\hdarg}{\tlarg}{\nil}{
\lettt{l'}{\cons{\xarg}{\cons{\hdarg}{\nil}}}
{\cons{l'}{\pairs(\xarg,\,\tlarg)}}}
\end{array}
$

\vspace{1.0em}

\noindent
$
\begin{array}{l}
\tf{\cprod}{\st{\alpha}{n}\times \st{\alpha}{m}}{\st{\st{\alpha}{2}}{n*m}}\\
\cprod(l_1, l_2)
\match{l_1}{\hdarg}{\tlarg}{\nil} {\append(\pairs(\hdarg, l_2),
\cprod(\tlarg, l_2))}
\end{array}
$

\vspace{1.0em}

\noindent
$
\begin{array}{l}
\tf{\sqdiff}{\st{\alpha}{n}\times \st{\alpha}{m}}
{\st{\st{\alpha}{2}}{(n^2+m^2-2*n*m)}}\\
\sqdiff(l_1, \ l_2) =
\match{l_1}{\hdarg}{\tlarg} {\cprod(l_2, \, l_2)} {
\match{l_2}{\hdarg'}{\tlarg'} {\cprod(l_1, \, l_1)} {
\sqdiff\ (\tlarg,\, \tlarg')}}
\end{array}
$

\vspace{1.0em}

For \textit{total} functions the following condition
is necessary:
\textit{for all instantiations * of size variables with themselves or zeros,
the inclusion  $\FVS{*\typesv_{n+1}} \subseteq \FVS{*\typesvs_1} \cup \cdots \cup
\FVS{*\typesvs_n}$ holds}. Consider, for instance, the first-order type
$\st{\st{\alpha}{m}}{n}  \rightarrow \st{\st{\alpha}{n}}{m}$, where on $\nil$ input, i.e. with
$n=0$, the input type degenerates to $\st{\st{\alpha}{m}}{0}\equiv \st{\st{\alpha}{0}}{0}$ but
the outer list of the output must have length $m$. This $m$ becomes unknown being
``hidden'' in $\st{\st{\alpha}{m}}{0}$. Thus,  this first-order type
may be accepted without the condition above, once
a function of this type is \textit{partial and undefined on empty lists}.
Since the type $\st{\st{\alpha}{m}}{n}  \rightarrow \st{\st{\alpha}{n}}{m}$ may
be assigned  to an implementation of
$n\times m$-matrix transposition,  undefinedness on $\nil$
may be interpreted as an exception ``cannot transpose an empty matrix''.

A context \context{} is a mapping from zero-order variables to
zero-order types. A signature \sign{} is a mapping from function
names to first-order types. The definition of $\FVS{-}$ is
straightforwardly extended to contexts.

\subsection{Typing rules}
\label{typing rules}

A typing judgement is a relation of the form $\judgement{D;\
\Gamma}{e}{\typesv}$, where $D$ is a conjunction
of equations between polynomials. $D$
is used to keep track of size information. In the current
language, the only place where size information is available is in
the nil-branch of the match-rule. The signature $\Sigma$ contains
the type assumptions for the functions that are called in the
expression under consideration. The typing judgement
relation is defined by the following rules:

$$
\begin{array}{l@{\quad \quad}l}
\infer[\textsc{IConst}]{\judgement{D;\ \Gamma}{c}{\intt}}{
\begin{array}{l}
\end{array}
} & \infer[\textsc{IBinop}]{\judgement{D;\ \Gamma,\
\tv{x}{\intt},\,\tv{y}{\intt}}{x \,\mathsf{binop} \, y }{\intt}}
{ \begin{array}{l} \end{array}}
\end{array}
$$

$$
\begin{array}{l@{\quad \quad}l}
\infer[\textsc{Nil}]{\judgement{D;\ \Gamma}{\nil}{\st{\typesv}{\sz}}}
{D \vdash p=0 }
 & \infer[\textsc{Var}]{\judgement{D;\ \Gamma,\
\tv{z}{\typesv}}{z}{\typesv{}'}}{ D \vdash \typesv = \typesv{}'}
\end{array}
$$

$$
\infer[\textsc{Cons}]{\judgement{D;\ \Gamma, \ \tv{\hdarg}{\typesv},\ \tv{\tlarg}{\st{\typesv}{\sz'}}}
{\cons{\hdarg}{\tlarg}}{\st{\typesv}{\sz}}}
{D \vdash \sz=\sz'+1}$$

$$
\infer[\textsc{If}] {\judgement{D;\ \Gamma,\ \tv{x}{\intt}} {\ifte{x}{e_{t}}{e_{f}}}
{\typesv}} {
\begin{array}{l}
\judgement{D;\ \Gamma,\ \tv{x}{\intt}}{e_{t}}{\typesv}\\
\judgement{D;\ \Gamma,\ \tv{x}{\intt}}{e_{f}}{\typesv}\\
\end{array}
}
$$

$$
\infer[\textsc{Let}]{\judgement{D;\ \Gamma}{\lete{z}{e_1}{e_2}}{\typesv}}
{
\begin{array}{c}
 z \notin \dom{\Gamma}\\
\judgement{D;\ \Gamma}{e_1}{\typesv_z}\\
\judgement{D;\ \Gamma , \ \tv{z}{\typesv_z}}{e_2}{\typesv}
\end{array}
}$$

$$
\infer[\textsc{Match}]
{\judgement{D;\ \Gamma, \
\tv{l}{\st{\typesv'}{\sz}}}{\match{l}{\hdarg}{\tlarg}{e_{\nil}}{e_{\mathsf{cons}}}}{\typesv}}
{\begin{array}{c}
\judgement{\sz=0,\ D;\ \Gamma, \ \tv{l}{\st{\typesv'}{p}}}{e_{\nil}}{\typesv} \\
\hdarg, \tlarg \not\in \dom{\Gamma} \quad \quad\judgement{D;\
\Gamma, \tv{\hdarg}{\typesv'}, \ \tv{l}{\st{\typesv'}{p}}, \
\tv{\tlarg}{\st{\typesv'}{\sz-1}}}{e_{\mathsf{cons}}}{\typesv}
\end{array}
}$$

The rule \textsc{LetFun} demands that all  $\mathsf{letfun}$-defined
functions, including recursive ones,
must be in the domain of the signature, and the corresponding first-order type must pass
type-checking:

$$\infer[\textsc{LetFun}]{\judgement{D;\ \Gamma}{\letfune{f}{z_1, \ldots, z_n}{e_1}{e_2}}{\typesv{}'}}
{
\begin{array}{c}
\sign(f) = \typesvs_1 \times \cdots \times \typesvs_n \rightarrow
\typesv_{n+1}\\
 \judgements{\ttt;\ \tv{z_1}{\typesvs_1}, \dots,
\tv{z_n}{\typesvs_n}}{e_1}{\typesv_{n+1}}{\sign}\\
\judgements{D;\ \Gamma}{e_2}{\typesv{}'}{\sign}
\end{array}
}$$

However, in practice we do not prohibit calls to functions
that are not defined via $\mathsf{letfun}$. If a function coming from a trusty external source
together with its first-order type  is declared via
$\mathsf{letextern}$, one applies the \textsc{LetExtern}
rule:

$$\infer[\textsc{LetExtern}]{\judgement{D;\ \Gamma}{\letexterne{f}{z_1, \ldots, z_n}{e}}{\typesv{}'}}
{
\begin{array}{c}
\sign(f) = \typesvs_1 \times \cdots \times \typesvs_n \rightarrow
\typesv_{n+1}\\
\judgements{D;\ \Gamma}{e}{\typesv{}'}{\sign}
\end{array}
}$$

\noindent When proving soundness we require  all
functions to be defined via $\mathsf{letfun}$
within an expression under consideration.

In the \textsc{FunApp}-rule, $\Theta$ computes the
substitution $*$ from its first argument (whose size expressions are always variables since they are taken from the first-order signature of the function) to its second
argument, and  the set $C$ of equations over size expressions from
$\typesv_1{}' \times \dots \times \typesv'_k{}$. The set
$C$ contains $p=p'$ if and only if
the expressions  $p$ and $p'$ are substituted to the same size variable. For instance,
if a function $\mathsf{dotprod}:\,\st{\intt}{m}\times\st{\intt}{m}\rightarrow  \intt$
is called with actual parameters of the types $\st{\intt}{n+n'+2}$ and $\st{\intt}{n+3}$,
then $C$ contains the equation $n+n'+2 = n+3$.

$$\infer[\textsc{FunApp}]
{ \judgement{D;\ \Gamma, \tv{z_1}{\typesv_1{}'}, \ldots,
\tv{z_n}{\typesv_{n}{}'}}{f(z_1, \ldots, z_k)}{\typesv_{n+1}{}'}}{
\begin{array}{c}
\langle*, C\rangle = \Theta(\typesvs_1{} \times \dots \times
\typesvs_n{},
\typesv_1{}' \times \dots \times \typesv_n{}')\\
\Sigma(f) = \typesvs_1 \times \ldots \times \typesvs_n \rightarrow
\typesv_{n + 1}
\quad \quad D \vdash \typesv{}'_{n+1}  = *(\typesv{}_{n+1})
\quad \quad D \vdash C
\end{array}}
$$
\noindent In the example with the call of $\mathsf{dotprod}$ the
equation $n+n'+2 = n+3$ holds if $D$ contains $n'-1=0$.

As another example of the \textsc{FunApp}-rule consider the recursive call
$\append(\tlarg,\,l_2)$
in the definition of $\append$:

$$\infer[\textsc{FunApp}]
{ \judgement
{\tv{\tlarg}{\st{\alpha}{n-1}},\, \tv{l_2}{\st{\alpha}{m}}}
{\append(\tlarg, \, l_2)}{\tau}}
{
\begin{array}{c}
\Sigma(\append) = \st{\alpha}{n} \times \st{\alpha}{m} \rightarrow \st{\alpha}{n+m}\\
\vdash \tau  = *(\st{\alpha}{n+m})
\end{array}}
$$
\noindent Here $\Theta(\st{\alpha}{n} \times  \st{\alpha}{m},\ \st{\alpha}{n-1} \times  \st{\alpha}{m})=\langle*,\,\emptyset \rangle$ with
$*(n)=n-1$, $*(m)=m$. Thus, $\tau  = *(\st{\alpha}{n+m}) = \st{\alpha}{n-1+ m}$.

The type system needs no conditions on non-negativity of size
expressions. Size expressions in types of meaningful data structures
are always non-negative. The soundness of the type system
ensures that this property is preserved throughout
(the evaluation of) a well-typed expression.

See subsection \ref{typechecking:examples} for examples of type checking in detail.

\subsection{Soundness of the type system}
\label{soundness}

Informally, soundness of the type system ensures that ``well-typed
programs will not go wrong''. This means that
if function arguments have meaningful values according to their types then
the result will have a meaningful value of the output type.
In section \ref{zero-order types}, we formalized the notion of a
meaningful value using a heap-aware semantics of types. Here we give an
operational semantics of the language.

We introduce a \emph{frame store} as a mapping from program variables to program values.
This mapping is maintained when a function body is evaluated. Before
evaluation of the function body starts, the store contains only the
actual parameters of the function. During evaluation, the store
is extended with the variables introduced by pattern matching or
$\mathsf{let}$-constructs. These variables are eventually bound to
the actual parameters, thus there is no access beyond the
current frame. Formally, a frame store is a finite partial map from
variables to values:

$$\St\ s\ :\ \expvars \rightharpoonup \Val$$

Using heaps and frame stores, and maintaining a mapping $\closures$
from function names to the bodies of the function definitions,
and a mapping $\closExt$ of external
function names to the external implementations, the operational semantics of expressions is defined by
the following rules:

$$
\infer[\textsc{OSIConst}]{\opsem{s}{h;\ \closures,\ \closExt}{c}{c}{h}}{c \in \intt}
$$

$$
\infer[\textsc{OSIBinop}]{\opsem{s}{h;\
\closures,\ \closExt}{\xarg \ \mathsf{binop}\,\yarg}{s(\xarg)\mathsf{binop} \,s(\yarg)}{h}}{}
$$

$$
\begin{array}{l@{\quad\ \quad}l}
\infer[\textsc{OSNil}]{\opsem{s}{h;\ \closures,\ \closExt}{\nil}{\nul}{h}}{} &
\infer[\textsc{OSVar}]{\opsem{s}{h;\
\closures,\ \closExt}{z}{s(z)}{h}}{}
\end{array}
$$

$$
\infer[\textsc{OSCons}]
{\opsem{s}{h,\ \closures, \ \closExt}{\cons{\hdarg}{\tlarg}}{\ell}{h[\ell.\hd:=v_{\hd},\
\ell.\tl:=v_{\tl}]}} {s(\hdarg)=v_{\hd} \quad\quad
s(\tlarg)=v_{\tl}\quad\quad \ell\notin\dom{h}}
$$

$$
\infer[\textsc{OSIfTrue}]{\opsem{s}{h;\ \closures,\ \closExt}{\ifte{x}{e_1}{e_2}}{v}{h'}}{s(x) \not= 0 \quad\quad \opsem{s}{h;\ \closures,\ \closExt}{e_1}{v}{h'}}$$

$$\infer[\textsc{OSIfFalse}]{\opsem{s}{h;\ \closures,\ \closExt}{\ifte{x}{e_1}{e_2}}{v}{h'}}{s(x) = 0 \quad\quad \opsem{s}{h;\ \closures,\ \closExt}{e_2}{v}{h'}}$$

$$
\infer[\textsc{OSLet}] {\opsem{s}{h;\
\closures,\ \closExt}{\lete{z}{e_1}{e_2}}{v}{h'}} {
\begin{array}{c}
\opsem{s}{h;\ \closures,\ \closExt}{e_1}{v_1}{h_1} \quad \quad
\opsem{s[z:=v_1]}{h_1;\ \closures,\ \closExt}{e_2}{v}{h'}
\end{array}
}
$$

$$
\infer[\textsc{OSMatch-Nil}]
{\opsem{s}{h;\ \closures,\ \closExt}{\match{l}{\hdarg}{\tlarg}{e_1}{e_2}}{v}{h'}}{s(l) = \nul \quad \quad \opsem{s}{h;\ \closures,\ \closExt}{e_1}{v}{h'}}$$

$$
\infer[\textsc{OSMatch-Cons}]{\opsem{s}{h;\ \closures,\ \closExt}{\match{l}{\hdarg}{\tlarg}{e_1}{e_2}}{v}{h'}}{\begin{array}{@{}c@{}}h.s(l).\hd = v_{\hd} \quad \quad h.s(l).\tl = v_{\tl}
\\ \opsem{s[\hdarg := v_{\hd}, \tlarg :=
v_{\tl}]}{h,\ \closures, \ \closExt}{e_2}{v}{h'} \end{array}}$$

$$\infer[\textsc{OSLetFun}]
{\opsem{s}{h;\ \closures,\ \closExt}{\letfune{f}{z_1, \ldots, z_n}{e_1}{e_2}}{v}{h'}}{
\begin{array}{c}
\opsem{s}{h;\ \closures{}[f := ((z_1, \ldots, z_n) \times e_1)],\ \closExt}{e_2}{v}{h'}\\
\FV{e_1} \subseteq \{z_1,\,\ldots,\,z_n\}
\end{array}
}$$


$$\infer[\textsc{OSFunApp}]{\opsem{s}{h;\ \closures,\ \closExt}{f(z_1, \ldots, z_n)}{v}{h'}}{
\begin{array}{@{}c@{}}
s(z_1) = v_1 \  \ldots \  s(z_n) = v_n \quad \quad \closures{}(f) = (z'_1, \ldots, z'_n) \times e_f \\
\opsem{[z'_1 := v_1, \ldots, z'_n := v_n]}{h;\ \closures,\ \closExt}{e_f}{v}{h'}\\
\FV{e_f} \subseteq \{z'_1,\,\ldots,\,z'_n\}
\end{array}}  $$


The soundness statement is defined by means of the following two
predicates. One indicates if a program value is meaningful with
respect to a certain heap and a ground type. The other does the same
for sets of values and types, taken from a frame store and a ground context
$\Gamma^{\,\bullet}$, respectively:

$$\begin{array}{lll} \soundval{v}{\typesvg}{h} & = &
\exists_w [\ \mdl{v}{h}{\typesvg}{w}\ ] \\ \soundstore{\mathit{vars}}
{\Gamma^{\,\bullet}}{s}{h} & = & \forall_{z \in
\mathit{vars}} [\ \soundval{s(z)}{\Gamma^{\,\bullet}(z)}{h}\ ] \end{array}$$

Let a valuation $\val$ map size variables to concrete (natural)
sizes and an instantiation $\eta$ map type variables to ground types:

$$\begin{array}{lll}\mathit{Valuation}\ & \val & :\ \sizevars \rightarrow \Z\\
\mathit{Instantiation}\ & \eta & :\ \typevars \rightarrow
\typesvg\end{array}$$

When applied to a type, context, or size equation,
valuations (and instantiations) map all variables occurring in it to
their valuation (or instantiation) images.

Now, stating the soundness theorem is straightforward:
\begin{thm}[Soundness]
Let $\opsem{s}{h;\ [\ ],\ [\ ]}{e}{v}{h'}$ and all functions called in $e$
be defined in it via the let-fun construct. Then for any context $\Gamma$,
signature $\Sigma$ and type $\typesv$ such that $
\judgement{\mathsf{True};\ \Gamma}{e}{\typesv}$ is derivable in the
type system and for any size valuation $\val$ and type instantiation
$\eta$, it holds that if the store is meaningful w.r.t. the context
$\eta(\val(\Gamma))$ then the output value is meaningful w.r.t the
type $\eta(\val(\typesv))$:
\vspace{0.3cm}

$\begin{array}{l}\forall_{\eta, \val} [\
\soundstore{\FV{e}}{\eta(\val(\Gamma))}{s}{h} \implies
\soundval{v}{\eta(\val(\typesv))}{h'}\ ]\end{array}
$
\end{thm}

The theorem follows from the
 following general statement:

\begin{lemma}[Soundness]
For any $s$, $h$, $\closures$, $e$, $v$, $h'$,
a set of equations $D$,
a context $\Gamma$,
a signature $\Sigma$, a type $\typesv$,
a size valuation $\val$ and  a type instantiation $\eta$
such that
\begin{itemize}
\item $\opsem{s}{h;\ \closures, \,[\,]\,}{e}{v}{h'}$,
\item $\judgements{D;\ \Gamma}{e}{\typesv}{\Sigma}$ is derivable in the
type system and
all functions called in $e$ are declared via $\mathsf{letfun}$,
\end{itemize}
\noindent one has
\vspace{0.3cm}

$\begin{array}{l}\forall_{\eta, \val} [\
\val(D) \ \land \
\soundstore{\FV{e}}{\eta(\val(\Gamma))}{s}{h} \implies
\soundval{v}{\eta(\val(\typesv))}{h'}\ ]\end{array}
$

\end{lemma}

The proof is done by induction on the
size of the derivation tree for the operational-semantics judgement.
For the \textsc{let}-rule it relies on \textit{benign sharing} \cite{POPL03}
of data structures. With benign sharing, shared heap structures to be used
in the let-body
are not changed by the let-binding expression of $\mathsf{let}$.
To formalize the notion of benign sharing we introduce a function
\textit{footprint} $\mt{\R}{\Hp \times \Val}{\Ps(\Loc)}$, which computes the set of
locations accessible in a given heap from a given value:

$$
\begin{array}{l@{\ }c@{\ }l}
\Reg{h}{c}& = & \emptyset \\
\Reg{h}{\nul} & = & \emptyset \\
\Reg{h}{\ell} & = &
\left\{
\begin{array}{l}
\emptyset, \;\;
\mathit{if}\ \ell \notin\dom{h}\\
\{\ell\} \,\cup
\,\Reg{\rstrloc{h}{\ell}}{h.\ell.\hd}\,\cup
\,\Reg{\rstrloc{h}{\ell}}{h.\ell.\tl}, \;\mathit{if}\ \ell \in
\dom{h}
\end{array}
\right.
\end{array}
$$
\noindent where $f| _X$ denotes the restriction of a (partial) map
$f$ to a set $X$.

We extend $\R$ to stores by $\Reg{h}{s}= \bigcup_{z\in
\dom{s}}\Reg{h}{s(z)}$. So, the operational-semantics let-rule with benign sharing
looks as follows:

$$
\infer[\textsc{OSLet}] {\opsem{s}{h;\
\closures,\, \closExt}{\lete{z}{e_1}{e_2}}{v}{h'}} {
\begin{array}{l}
\opsem{s}{h;\ \closures,\, \closExt}{e_1}{v_1}{h_1}\\
\opsem{s[z:=v_1]}{h_1;\ \closures,\, \closExt}{e_2}{v}{h'}\\
h | _{\Reg{h}{s|_{\FV{e_2}}}} =  h_1 | _{\Reg{h}{s|_{\FV{e_2}}}} \\
\end{array}
}
$$

This semantic condition is not statically typable in general, however,
there are type systems that approximate it,
e.g. linear typing and uniqueness typing \cite{Uniq}.
Since in our language we have neither destructive pattern matching nor assignments,
benign sharing is guaranteed.

\begin{proof}
Let everywhere below $\opsem{s}{h;\ \closures}{e}{v}{h'}$
denote the operational-semantics judgement
$\opsem{s}{h;\ \closures, \,[\,]\,}{e}{v}{h'}$ with
the empty external closure.

In the proof we will use a few technical lemmata about heaps and model relations.
They are intuitively clear statements like ``extending a heap does not change a model relation'', so we do not prove them in
the main part of the paper. The interested reader may find the technical proofs in the appendix.

For the sake of convenience we will denote
$\eta(\val(\typesv))$ via $\typesv_{\eta\val}$, 
$\eta(\val(\Gamma))$ via
$\Gamma_{\eta\val}$ and $\val(D)$ via $D_{\val}$.

We prove the statement by induction on the height of the derivation
tree for the operational-semantics judgement. Given $\opsem{s}{h;\
\closures}{e}{v}{h'}$ fix some $\Gamma$, $\Sigma$, and $\typesv{}$, such
that $\judgements{D; \ \Gamma}{e}{\typesv{}}{\Sigma}$. Fix a valuation
$\val \in \FV{\Gamma} \cup \FV{\typesv{}} \rightarrow \Z$,
a type instantiation $\eta \in \FV{\Gamma} \cup \FV{\typesv{}} \rightarrow \typesvg{}$,
such that $D_{\val}$ and $\soundstore{\FV{e}}{\Gamma_{\eta\val}}{s}{h}$ hold.
We must show that $\soundval{v}{\typesv{}_{\eta\val}}{h'}$ holds.

\begin{description}

\item[OSIConst] In this case $v = c$ for some constant $c$ and $\typesv{} = \intt$.
Then, by the definition we have $\mdl{c}{h}{\intt}{c}$ and $\soundval{v}{\intt}{h'}$.

\item[OSNull] In this case $v = \nul$ and $\typesv{} = \st{\typesv{}'}{0}$ for some $\typesv{}'$.
Then, by the definition we have
$\mdl{\nul}{h}{\st{\typesv'_{\eta\val}}{0}}{\mbox{\texttt{[]}}}$.

\item[OSVar] From $D \vdash \tau = \tau'$ and $D_\val$
it follows that $\typesv_{\eta\val}=\typesv'_{\eta\val}$.
From this and
$$\soundstore{\FV{z}}{\Gamma\cup({z:\typesv'})_{\eta\val}}{h}{s}$$
\noindent it follows that

$$\soundval{s(z)}{\typesv{}_{\eta\val}}{h}$$

\item[OSCons]
In this case $e = \cons{\hdarg}{\tlarg}$, $\typesv{} = \st{\typesv{}'}{p}$,
$\{ \tv{\hdarg}{\typesv{}'}, \tv{\tlarg}{\st{\typesv{}'}{p'}} \} \subseteq \Gamma$ for some $\hdarg$, $\tlarg$, $p'$ and $\typesv{}'$. Since $\soundstore{\FV{e}}{\Gamma_{\eta\val}}{s}{h}$ there exist $w_\hd$ and $w_\tl$ such that
$\mdl{s(\hdarg)}{h}{\typesv'_{\eta\val}}{w_\hd}$ and
$\mdl{s(\tlarg)}{h}{(\st{\typesv'}{p'})_{\eta\val}}{w_\tl}$. From the
operational semantics judgement we have that $v = \ell$ for some
location $\ell \notin \dom{h}$, and $h'=h[\ell.\hd:=s(\hdarg),\
\ell.\tl:=s(\tlarg)]$.  Therefore,
$\mdl{h'.\ell.\hd}{h}{\typesv{}'_{\eta\val}}{w_\hd}$ and
$\mdl{h'.\ell.\tl}{h}{(\st{\typesv{}'}{p'})_{\eta\val}}{w_\tl}$ hold as well.
It is easy to see that $h=h'|_{\dom{h'}\setminus\{\ell\}}$.

Thus,

$$
\begin{array}{l}
\mdl{h'.\ell.\hd}{h'|_{\dom{h'}\setminus\{\ell\}}}{\typesv{}'_{\eta\val}}{w_\hd}\\
\mdl{h'.\ell.\tl}{h'|_{\dom{h'}\setminus\{\ell\}}}{(\st{\typesv{}'}{p'})_{\eta\val}}{w_\tl}
\end{array}
$$
\noindent
This and $D_\val$, which implies $p_\val=(p'+1)_\val$ gives
$\mdl{\ell}{h'}{(\st{\typesv{}'}{p})_{\eta\val}}{w_\hd::w_\tl}$ and
thus $\soundval{\ell}{\typesv{}_{\eta\val}}{h'}$.

\item[OSIfTrue] In this case $e = \ifte{x}{e_1}{e_2}$ for some $e_1$, $e_2$, and $x$. Knowing that
$\judgement{D; \  \Gamma}{e_1}{\typesv{}}$ we apply the induction
hypothesis to the derivation of $\opsem{s}{h;\
\closures}{e_1}{v}{h'}$, with the same $\eta$, $\val$ to obtain
$\soundstore{\FV{e_1}}{\Gamma_{\eta\val}}{s}{x} \implies
\soundval{v}{\typesv{}_{\eta\val}}{h'}$. From $\FV{e_1} \subseteq \FV{e}$,
$\soundstore{\FV{e}}{\Gamma_{\eta\val}}{s}{x}$, and lemma \ref{LEMMA:SubsetFV}
it follows that $\soundval{v}{\typesv{}_{\eta\val}}{h'}$.

\item[OSIfFalse] is similar to the true-branch.

\item[OSLetFun] The result follows from the induction hypothesis
for $$\opsem{s}{h;\ \closures[f:=(\bar{z} \times e_1)]}{e_2}{v}{h'},$$
\noindent
with $\judgements{D;\ \Gamma}{e_2}{\typesv}{\Sigma}$ and the same $\eta$, $\val$,
store $s$ and heap $h$.

\item[OSLet] In this case $e = \lete{z}{e_1}{e_2}$ for some $z$, $e_1$, and $e_2$ and we have $\opsem{s}{h;\ \closures}{e_1}{v_1}{h_1}$ and $\opsem{s[z:=v_1]}{h_1;\ \closures}{e_2}{v}{h'}$ for some $v_1$ and $h_1$.
We know that $\judgement{D; \  \Gamma}{e_1}{\typesv{}'}$, $z \not\in
\Gamma$ and $\judgement{D; \  \Gamma,
\tv{z}{\typesv{}'}}{e_2}{\typesv{}}$ for some $\typesv{}'$. Applying
the induction hypothesis to the first branch gives
$\soundstore{\FV{e_1}}{\Gamma_{\eta\val}}{s}{h} \implies
\soundval{v_1}{\typesv{}'_{\eta\val}}{h_1}$. Since $\FV{e_1} \subseteq
\FV{e_1} \cup (\FV{e_2} \setminus \{z\}) = \FV{e}$ and
$$\soundstore{\FV{e}}{\Gamma_{\eta\val}}{s}{h}$$
\noindent we have from lemma
\ref{LEMMA:SubsetFV} that \soundstore{\FV{e_1}}{\Gamma_{\eta\val}}{s}{h} holds
and hence we have \soundval{v_1}{\typesv{}'_{\eta\val}}{h_1}.

Now apply the induction hypothesis to the second branch to get
$$\soundstore{\FV{e_2}}{\Gamma_{\eta\val} \cup \{ \tv{z}{\typesv{}'_\val}
\}}{s[z := v_1]}{h_1} \implies \soundval{v}{\typesv{}_{\eta\val}}{h'}.$$
\noindent

Now we will show that the l.h.s. of the implication holds.
Fix some $z' \in \FV{e_2}$. If $z' = z$, then
\soundval{v_1}{\typesv{}'_{\eta\val}}{h_1} implies
\soundval{s[z:=v_1](z)}{\typesv{}'_{\eta\val}}{h_1}. If $z' \ne z$, then
$s[z:= v_1](z') = s(z')$. Because we know that sharing is benign,
$h|_{\Reg{h}{s(z')}}=h_1|_{\Reg{h}{s(z')}}$, applying lemma
\ref{LEMMA:EqFootprModels} and then \ref{LEMMA:SubsetFV} we have that
$\mdl{s(z')}{h}{\Gamma_{\eta\val}(z')}{w_{z'}}$ implies
$\mdl{s(z')}{h_1}{\Gamma_{\eta\val}(z')}{w_{z'}}$ implies
$\mdl{s[z:=v_1](z')}{h_1}{\Gamma_{\eta\val}(z')}{w_{z'}}$ and
thus \soundval{s[z:=v_1](z')}{\Gamma_{\eta\val}(z')}{h_1}. Hence,
$\soundstore{\FV{e_2}}{\Gamma _{\eta\val}\cup \{\tv{z}{\typesv{}'}_{\eta\val}
\}}{s[z:=v_1]}{h_1}$. Therefore, \soundval{v}{\typesv{}_{\eta\val}}{h'}.

\item[OSMatch-Nil] In this case $e = \matchline{l}{\hdarg}{\tlarg}{e_1}{e_2}$
for some $l$, $\hdarg$, $\tlarg$, $e_1$, and $e_2$.
The typing context has the form
$\Gamma=\Gamma'\cup\{\tv{l}{\st{\typesv{}'}{p}}\}$ for some
$\Gamma'$, $\typesv{}'$, $p$. The operational-semantics derivation
gives $s(l) = \mathtt{NULL}$, hence validity for $s(l)$ gives
$\tv{l}{\st{\typesv{}'}{0}}$ and thus $\val(p)=0$. From the typing
derivation for $\judgement{D; \  \Gamma}{e}{\typesv{}}$ we then know
that $\judgement{p=0,\,D; \  \Gamma'}{e_1}{\typesv{}}$. Applying the induction
hypothesis, with $p=0\land D$ then yields $\soundstore{\FV{e_1}}{\Gamma'_{\eta\val}}{s}{h}
\implies \soundval{v}{\typesv{}_{\eta\val}}{h'}$. From $\FV{e_1} \subseteq
\FV{e}$, \soundstore{\FV{e}}{\Gamma_{\eta\val}}{s}{h}, $\val(p)=0 \land D_{\val}$ and lemma
\ref{LEMMA:SubsetFV} it follows that $\soundval{v}{\typesv{}_{\eta\val}}{h'}$.

\item[OSMatch-Cons] In this case
$e = \matchline{l}{\hdarg}{\tlarg}{e_1}{e_2}$ for some $l$, $\hdarg$, $\tlarg$, $e_1$, $e_2$.
The typing context has the form
$\Gamma=\Gamma'\cup\{\tv{l}{\st{\typesv{}'}{p}}\}$ for some
$\Gamma'$, $\typesv{}'$, $p$. From the operational semantics we know
that $h.s(l).\hd = v_\hd$ and $h.s(l).v_\tl$ for some $v_\hd$ and
$v_\tl$ -- that is $s(l)\ne \nul$ -- hence, due to validity of
$s(l)$, we have $\tv{l}{\st{\typesv{}'}{p}}$ for some $\typesv{}'$
and $\val(p) \geq 1$. From the typing derivation of $e$ we obtain
that $\judgement{D; \  \Gamma', \
\tv{l}{\st{\typesv{}'}{p}},\ \tv{\hdarg}{\typesv{}'}, \
\tv{\tlarg}{\st{\typesv{}'}{p - 1}}}{e_2}{\typesv{}}$
 Applying the induction hypothesis yields
$$
\begin{array}{l}
\soundstore{\FV{e_2}}{
\left\{
\begin{array}{l}
\Gamma'_{\eta\val} \cup \\
\cup\{\tv{l}{(\st{\typesv{}'}{p})_{\eta\val}}\}\cup\\
\cup\{\tv{\hdarg}{\typesv{}'_{\eta\val}}\}\cup\\
\cup\{\tv{\tlarg}{\st{\typesv{}'}{e}}\}_{\eta\val}\}
\end{array}
\right\}
}
{s
\left[
\begin{array}{l}
\hdarg:=v_\hd, \\
\tlarg:=v_\tl
\end{array}
\right]
}{h} \implies \\
 \implies \soundval{v}{\typesv{}_{\eta\val}}{h'}.
\end{array}
$$
\noindent
Show that the l.h.s. of the implication holds.
From $\soundstore{\FV{e}}{\Gamma_{\eta\val}}{s}{h}$,
$(\FV{e_2}\setminus \{\hdarg,\ \tlarg\}) \subseteq \FV{e}$, and lemma \ref{LEMMA:SubsetFV}
we obtain $$\soundstore{\FV{e_2}\setminus \{\hdarg,\ \tlarg\}}{\Gamma_{\eta\val}}{s}{h}$$
Due to $\hdarg, \tlarg \not\in \dom{s}$
we can apply lemma \ref{LEMMA:ChangeStore} and get
$$\soundstore{\FV{e_2}\setminus\{\hdarg,\ \tlarg\}}{\Gamma_\val}{s[\hdarg:=v_\hd, \tlarg:=v_\tl]}{h}$$

From the validity
$\mdl{s(l)}{h}{(\st{\typesv'}{p})_{\eta\val}}{w_\hd::w_\tl}$, and
obvious $\val(p-1)=\val(p)-1$ the validity
of $v_\hd$ and $v_\tl$ follows:
$\mdl{v_\hd}{h}{\typesv'_{\eta\val}}{w_\hd}$,
$\mdl{v_\tl}{h}{(\st{\typesv'}{p-1})_{\eta\val}}{w_\tl}$.

Now
$\soundstore{\FV{e_2}}{\Gamma_{\eta\val}  \cup \{ \tv{\hdarg}{\typesv'},
\tv{\tlarg}{\st{\typesv'}{p - 1}} \}_{\eta\val}}{s[\hdarg:=v_\hd,\ \tlarg:=v_\tl]}{h}$  and, hence,
$$\soundval{v}{\typesv{}_{\eta\val}}{h'}.$$

\item[OSFunApp]

We want to apply the induction assumption to
$$\opsem{[z'_1 := v_1, \ldots, z'_n := v_n]}{h;\ \closures}{e_f}{v}{h'}.$$

Let $\Sigma(f)=\typesvs_1 \times \ldots \times \typesvs_n \rightarrow \tau'$,
the types $\typesvs_i$ of the formal parameters be $\st{\ldots \st{\alpha_i}{n_{ik_i}}\ldots} {n_{i1}}$ respectively,
and the types $\Gamma(z_i)$ of the actual parameters $z_i$ be $\st{\ldots \st{\tau_{\alpha_i}}{p_{ik_i}}\ldots }{p_{i1}}$,
where $1 \leq i \leq n$.
According to the typing rule
$D \vdash \tau = \tau'[\ldots \alpha_i:=\tau_{\alpha_i}\ldots]\,[\ldots n_{ij}:=p_{ij}\ldots]$.

Since all called in $e$ functions are defined via $\mathsf{letfun}$,
there must be a node in the derivation tree
with $\judgements{\ttt,\ z'_1:\typesvs_1,\ldots,\,z'_n:\typesvs_n}{e_f}{\typesv'}{\Sigma}$.

We take $\eta'$ and $\val'$, such that
\begin{itemize}
\item $\eta'(\alpha_i)=\eta(\typesv_{\alpha_i})$,
\item  $\val'(n_{ij})=\val(p_{ij})$.
\end{itemize}
\noindent Thus, $\Gamma(z_i)_{\eta\val}=(\typesvs_i)_{\eta'\val'}$, since
{\small
$$(\typesvs_i)_{\eta'\val'}=\st{\ldots \st{\eta'(\alpha_i)}{\val'(n_{ik_i})}\ldots} {\val'(n_{i1})}=
\st{\ldots \st{\eta(\tau_{\alpha_i}}{\val(p_{ik_i})}\ldots} {\val(p_{i1})}=
(\Gamma(z_i))_{\eta\val}$$
}

$\ttt$ (``no conditions'') holds trivially on $\val'$. From the induction assumption we have

$$
\begin{array}{l}
\soundstore{(z'_1,\ldots\,z'_n)}
{(z'_1:\typesvs_{1\,\eta'\val'},\ldots,\,z'_n:\typesvs_{n,\ \eta'\val'})}
{[z'_1 := v_1,\ldots, z'_n := v_n]}{h} \\
\implies
\soundval{v}{\typesv'_{\eta'\val'}}{h'}
\end{array}
$$

Show that the l.h.s. holds.
From $\soundstore{\FV{e}}{\Gamma_{\eta\val}}{s}{h}$ we have validity of
the values of the actual parameters:
$\mdl{v_i}{h}{\Gamma_{\eta\val}(z_i)}{w_i}$ for some $w_i$, where $1\leq i
\leq k$. Since $\Gamma_{\eta\val}(z_i)=(\typesvs_i)_{\eta'\val'}$,
the left-hand side of the implication holds,
and one obtains $\soundval{v}{\typesv'_{\eta'\val'}}{h'}$.

Now,  $D_\val$  implies $
\typesv_{\eta\val}=
\typesv'[\ldots\alpha_i:=\typesv_{\alpha_i}\ldots][\ldots n_{ij}:=p_{ij}\ldots]_{\eta\val}$.
Then from the construction for $\eta'$ and $\val'$ it follows
$\typesv'[\ldots\alpha_i:=\typesv_{\alpha_i} \ldots][\ldots n_{ij}:=p_{ij}\ldots]_{\eta\val}=
\typesv'[\ldots\alpha_i:=\eta(\typesv_{\alpha_i})\ldots][\ldots n_{ij}:=\val(p_{ij})\ldots]=
\typesv'_{\eta'\val'}$

Thus, we have $\soundval{v}{\typesv_{\eta\val}}{h'}$.

\end{description}
\end{proof}

\subsection{Completeness of the type system}
\label{completeness}

Recall, that the system we consider is constituted from zero- and first-order types, typing rules,
and Peano arithmetic extended to rationals.

The system is not complete in the class of shapely function definitions: there are shapely functions for which
shapeliness may not be proved by means of the typing rules and the arithmetic. In other words,
their annotated type cannot be checked by the system.  For instance consider the following
expression $e$:

$$
\begin{array}{l}
 \mathsf{let}\,l =f(z_1,\ldots,\,z_k)\;\mathsf{in}\\
\quad \quad \lete{x}{\fapp{\mathsf{length}}{l}}{\ifte{x}{\cons{1}{\nil}}{\nil}}\\
\end{array}
$$
\vspace{0.1cm}
\noindent
where $\mathsf{length}(x)$ returns the length of list $x$.
Let $p_f(n_1,\ldots,\,n_k)$ denote the polynomial size dependency
for the shapely function definition $f$. If $f$ never outputs an empty list, then
the expression  $e$ defines a shapely function, with a polynomial size dependency
$p(n_1,\ldots,\,n_k)=1$. Otherwise $p(n_1,\ldots,\,n_k)=0$ when $f$ outputs $\nil$.
Suppose, there exists a procedure, that for any instantiation of
the expression with $f$, produces its shapely type, when it is shapely, or
rejects it otherwise. Then this procedure is capable to solve
\textit{10th Hilbert problem}: whether there exists a general procedure that given a
polynomial with integer coefficients decides if this polynomial has
natural roots or not.\footnote{The original formulation is about
integer roots. However, both versions are equivalent and logicians
consider natural roots.} Matiyasevich \cite{Mat91} has shown that
such a procedure does not exist. A similar problem is connected with $\mathsf{match}$-construct.

We study constructions like above in more detail
in  section \ref{typechecking:undec}, devoted to decidability of type-checking.
In particular, in lemma \ref{LEMMA:sizepolyn}
we show, that for any integer polynomial $q$ there is a shapely function definition $f$
such that its size polynomial $p_f(n_1,\ldots,n_k)$ is equal to $q^2(n_1,\ldots,n_k)$
and thus $p_f$ has roots if and only if $q$ has roots.

In fact, this example shows that not only our system, but any system using
integer arithmetic, is not complete in the class of shapely function definitions.


\section{Type Checking}
\label{typechecking}

Because for every syntactic construction there is only one typing
rule that is applicable, type checking is straightforward.
The procedure parses a given function body and reduces to
proving equations for rational polynomials. Consider some examples.

\subsection{Examples}
\label{typechecking:examples}

\subsubsection{Cartesian product} In the
introduction, the Cartesian product was implemented using a
``sugared'' syntax. Here, we present the $\cprod$ function in
the language defined in section \ref{language}.
\vspace{0.1cm}

$\begin{array}{@{}l@{\ }l@{\ }l@{}}
\mathsf{letfun}\ \cprod(l_1, \,l_2) & = &
\match{l_1}{\hdarg}{\tlarg}{\nil}{\begin{array}[t]{@{}l@{\ }l@{\
}l@{}} \mathsf{let}\ l' & = & \fapp{\pairs}{\hdarg, \,l_2} \\
\mathsf{in\
let}\ l'' & = & \cprod(\tlarg, y) \\
\multicolumn{3}{@{}l@{}}{\mathsf{in}\ \fapp{\append}{l', \,l''}}\\
 \end{array}}\\
\mathsf{in} \, \ldots \\
\end{array}$
\vspace{0.1cm}

Functions $\pairs$ and $\append$ are assumed to be defined in the
core syntax of the language as well. Hence, $\Sigma$ contains the following types:
\vspace{0.1cm}

$
\begin{array}{lll}
\Sigma(\append) & = & \st{\alpha}{n} \times \st{\alpha}{m}
\rightarrow \st{\alpha}{n + m}\\
\Sigma(\pairs) & = & \alpha \times \st{\alpha}{m} \rightarrow
\st{\st{\alpha}{2}}{m} \\
\Sigma(\cprod) & = & \st{\alpha}{n} \times \st{\alpha}{m}
\rightarrow \st{\st{\alpha}{2}}{n * m}
\end{array}
$
\vspace{0.1cm}

\noindent To type-check $\cprod{} : \st{\alpha}{n} \times
\st{\alpha}{m} \rightarrow \st{\st{\alpha}{2}}{n * m}$ means to
check:
\vspace{0.1cm}

\begin{tabular}{l@{\quad}l} \textsc{Prove:} &
$\judgement{\tv{l_1}{\st{\alpha}{n}},
\tv{l_2}{\st{\alpha}{m}}}{e_{\cprod}}{\st{\st{\alpha}{2}}{n*m}},$
\end{tabular}
\vspace{0.1cm}

\noindent where $e_{\cprod}$ is the function body. This is demanded
by the first branch of the \textsc{LetFun}-rule. Applying
the \textsc{Match}-rule branches the proof:
\vspace{0.1cm}

\begin{tabular}{l@{\quad}l}
\textsc{Nil:} & $\judgement{n=0;\ \tv{l_2}{\st{\alpha}{m}}}{\nil}{\st{\st{\alpha}{2}}{n*m}}$\\
\end{tabular}

\begin{tabular}{l@{\quad}l}
\textsc{Cons:} &
$\begin{array}[t]{@{}l@{}}\judgement{\tv{\hdarg}{\alpha}, \
\tv{l_1}{\st{\alpha}{n}},\ \tv{\tlarg}{\st{\alpha}{n - 1}}, \
\tv{l_2}{\st{\alpha}{m}}}{\\
\left.
\begin{array}{@{\quad\quad\quad\quad\quad\quad\quad\quad\quad\quad\quad}l@{\ }l@{\ }l@{}}
\mathsf{let}\ l' & = & \fapp{\pairs}{\hdarg, l_2} \\
\mathsf{in\
let}\ l'' & = & \cprod(\tlarg, l_2) \\
\multicolumn{3}{@{\quad\quad\quad\quad\quad\quad\quad\quad\quad\quad\quad}l@{}}{\mathsf{in}\
\fapp{\append}{l', \,l''}}
 \end{array}
\right\}
}
{\st{\st{\alpha}{2}}{n * m}}\end{array}$
\end{tabular}
\vspace{0.1cm}

\noindent Applying the \textsc{Nil}-rule to the \textsc{Nil}-branch
gives $n=0 \vdash n * m = 0$, which is trivially true. The
\textsc{Cons}-branch is proved by applying the \textsc{Let}-rule
twice. This results in three proof
obligations:\\

\begin{tabular}{l@{\quad}l}
\textsc{Bind-l':} & $\judgement{\tv{\hdarg}{\alpha}, \
\tv{l_2}{\st{\alpha}{m}}}{\fapp{\pairs}{\hdarg, \,l_2}}{\tau_1}$\\
\textsc{Bind-l'':} & $\judgement{\tv{\tlarg}{\st{\alpha}{n - 1}}, \
\tv{l_2}{\st{\alpha}{m}}}{\fapp{\cprod}{\tlarg, \,l_2}}{\tau_2}$\\
\textsc{Body:} & $\judgement{\tv{l'}{\tau_1},
\tv{l''}{\tau_2}}{\fapp{\append}{l', l''}}{\st{\alpha}{n * m}}$
\end{tabular}\\

\noindent From the applications of the \textsc{FunApp}-rule to
\textsc{Bind-l'} and \textsc{Bind-l''} it follows that $\tau_1$
should be $\st{\st{\alpha}{2}}{m}$ and $\tau_2$ should be
$\st{\st{\alpha}{2}}{(n - 1)
* m}$. Lastly,
applying the \textsc{FunApp}-rule to  \textsc{Body} yields the proof
obligation $\vdash n*m = m+ (n - 1) * m$, which is true in the
axiomatics.

\subsubsection{Example with negative coefficients}

In contrast to the system presented by Vasconcelos and Hammond
\cite{Vas03}, where only subtraction of constants are allowed,
our system allows negative coefficients in size
expressions. Of course, this is only a valid size expression
(yielded by a total function) if the
polynomial maps naturals into naturals. Here, we
show an example where this is the case.  Given two lists, the
function ``subtracts'' elements from lists simultaneously, till one
of the lists is empty. Then, the Cartesian product of the remaining list
with itself is returned:
\vspace{0.1cm}

$\func{\sqdiff}
{(l_1, \ l_2)} { \match{l_1}{\hdarg}{\tlarg} {\cprod(l_2, \, l_2)} {
\match{l_2}{\hdarg'}{\tlarg'} {\cprod(l_1, \, l_1)} {
\sqdiff\ (\tlarg,\, \tlarg') } } }.$
\vspace{0.1cm}

\noindent
It can be checked that $\sqdiff$ has type
$\st{\alpha}{n}\times \st{\alpha}{m} \rightarrow
\st{\st{\alpha}{2}}{(n^2+m^2-2*n*m)}$.

\subsection{Type checking in general is undecidable (even for total function definitions)}
\label{typechecking:undec}

In the examples above, type checking ends up
with a set of entailments like $n=0\vdash n*m=0$ or $\vdash
n*m = m+m*(n-1)$ that have to hold. However, we show that there is no
procedure to check all possible entailments that may arise. To
make type checking decidable, we formulate a syntactical condition
on the structure of a program expression that ensures the
entailments have a trivial form. The condition  is as follows:
\textit{given a function body, allow pattern-matching only on the function parameters
or variables bound to them by other pattern-matchings}. Thus,
we prohibit expressions like

$
\lete{l}{f_0(x_1,\,\ldots, \,x_k)}
{\begin{array}[t]{l}
\match{l}{\hdarg}{\tlarg}{e_1}{e_2}
\end{array}
}
$

\noindent Pattern-matching like

$
\match{l}{\hdarg}{\tlarg}{e_1}
{\begin{array}[t]{l}
\match{\tlarg}{\hdarg'}{\tlarg'}{e'_1}{e_2}
\end{array}
}
$

\noindent is allowed. Below we explain the reason for this restriction.

We show that the existence of a procedure that checks all
possible entailments at the end of type checking
is reduced to Hilbert's tenth problem.  Type checking is reducible to
a procedure for checking if arbitrary size
polynomials of shapely functions have natural roots. It turns out
that the latter is the same as finding natural roots of integer
polynomials.

Consider the following expression $e_{H}$ with free variables $l_1,\,\ldots, \,l_k$:
\vspace{0.1cm}

$
\lete{l}{f_0(l_1,\,\ldots, \,l_k)}
{\begin{array}[t]{l}
\match{l}{\hdarg}{\tlarg}
{f_1(l_1,\,\ldots, \,l_k)}{f_2(l_1,\,\ldots, \,l_k)}
\end{array}
}
$
\vspace{0.1cm}

\noindent We check if it has the type $\st{\alpha_1}{n_1} \times
\ldots\times \st{\alpha_k}{n_k} \longrightarrow
\st{\alpha}{p(n_1,\ldots,\,n_k)}$, given that
$f_i:\st{\alpha_1}{n_1} \times \ldots\times \st{\alpha_k}{n_k}
\longrightarrow \st{\alpha}{p_i(n_1,\ldots,\,n_k)}$, with
$i=0,\,1,\,2$. Then at the end of the type checking procedure we
obtain the entailment:
\vspace{0.1cm}

$p_0(n_1,\ldots,\,n_k)=0\vdash p_1(n_1,\ldots,\,n_k)=p(n_1,\ldots,\,n_k).$
\vspace{0.1cm}

Even if $p$ and $p_1$ are not equal, say $p_1=0$ and $p=1$, it does
not mean that type checking fails; it might not be possible to enter
the ``bad'' nil-branch. To check if the nil-branch is entered means
to check if $p_0=0$ has a solution in natural numbers.
Thus,  a type-checker for any size polynomial $p_0$ must be able to
decide if it has natural roots or not.

Checking if any size polynomial has roots in natural numbers, is as difficult
as checking whether an arbitrary polynomial has roots or not.
First, we prove the following lemma.

\begin{lemma}
\label{LEMMA:sizepolyn}
For any polynomial $q$ there is a total shapely function definition $f$ such that
its size dependency $p_f(n_1,\ldots,\,n_k)$ is equal to $q^2(n_1,\ldots,\,n_k)$.
\end{lemma}
\proof
First, note that any polynomial $q$ may be presented as the
difference $q_1-q_2$ of two polynomials with non-negative
coefficients\footnote{If $q=\Sigma a_{i_1,\ldots,i_k}x_1^{i_1}\ldots
x^{i_k}_k$, then $q_1=\Sigma_{a_{i_1,\ldots,i_k}\geq 0}
a_{i_1,\ldots,i_k}x_1^{i_1}\ldots x^{i_k}_k$, and
$q_2=\Sigma_{a_{i_1,\ldots,i_k}< 0}
|a_{i_1,\ldots,i_k}|x_1^{i_1}\ldots x^{i_k}_k$.}. So,
$q^2=(q_1-q_2)^2$ is a size polynomial, obtained by superposition of
$\sqdiff$ with $q_1$ and $q_2$.  Here $q_1$ and $q_2$ are size
polynomials with positive coefficients for corresponding
compositions of $\append$ and
$\mathsf{copyfirst}\,:\,\st{\alpha}{n}\times\st{\alpha}{m} \rightarrow \st{\alpha}{n*m}$
(see subsection \ref{motivation})
functions.\qed

Summing up the constructions above we obtain the following statement:
\begin{lemma}
\label{LEMMA:undec}
If there exists a type-checker that for any function definition and its type annotation
is able to accept or reject the annotated type correctly, then there exists a procedure
that for any integer polynomial $q(n_1,\,\ldots,\,n_k)$ decides if it has natural roots
or not.
\end{lemma}
\proof
Suppose that such type checker exists. Consider the expression $e_H$ above with
$f_0$, $f_1$, $f_2$ defined as follows.
Using lemma \ref{LEMMA:sizepolyn}, construct
a function definition $f_0$ that has a size dependency $q^2(n_1,\,\ldots,\,n_k)$.
Now let $f_1$ be  defined by the expression $\nil$ and
let $f_2$ be  defined by $\cons{1}{\nil}$.

The type checker accepts $e_H$ with the type annotation
$p\equiv 1$ if and only if the $\nil$-branch is not entered, that is if and only if
$q^2(n_1,\,\ldots,\,n_k)$ has no roots. Trivially, $q^2(n_1,\,\ldots,\,n_k)$ has roots
if and only if $q(n_1,\,\ldots,\,n_k)$ does.\qed

So, existence of a general type-checker reduces to solving Hilbert's
tenth problem. Hence, type checking is undecidable.

We can show this in a more constructive way using the stronger form
of the undecidability of Hilbert's tenth problem:
for any type-checking procedure $\mathcal{I}$ one can construct
a program expression, for which $\mathcal{I}$ fails to give the correct answer.
We will use the result of Matiyasevich who has proved the following: there is a one-parameter
Diophantine equation $W(a, \,n_1,\ldots,\,n_k)=0$ and an algorithm which for given
algorithm $\mathcal{A}$ produces a number $a_\mathcal{A}$ such that
$\mathcal{A}$ fails to give the correct answer for the question
whether equation $W(a_\mathcal{A}, \,n_1,\ldots,\,n_k)=0$ has a solution in $(n_1,\ldots,\,n_k)$.
So, if in the example above one takes the function $f_0$ such that
its size polynomial $p_0$ is the square of the $W(a_\mathcal{I}, \,n_1,\ldots,\,n_k)$
and $p=1$, $p_1=0$, then the type checker $\mathcal{I}$ fails to give
the correct answer for $e_{H}$.

An anonymous reviewer pointed out  that the construction from lemma
\ref{LEMMA:sizepolyn} demonstrates a problem with
real arithmetic, when it is used to check numerical entailments, generated by the type checker.
Suppose we want to omit the syntactic restriction and type check the expression
$e_H$ where the size dependency for $f_0$ is $p_0(n)=(n^2-2)^2$. A real-arithmetic-based version
of the checker rejects $e_H$, since there is a real root for $p_0$ and
in this abstract interpretation the $\nil$-branch with $1=0$ must be considered.
In fact, the expression is well-typed with annotation $p\equiv 1$, since
there is no natural roots for $p_0$ and the $\nil$-branch is never entered.

For
checking
a particular expression it is sufficient to solve the
corresponding sets of Diophantine equations. Type checking depends
on decidability of Diophantine equations from $D$ in any entailment
$D\vdash p=p'$, where $p$ is not equal to $p'$ in general (but might
be if the equations from $D$ hold). If we have a solution for $D$ we
can substitute this solution in $p$ and $p'$. If a solution over
variables $n_1,\ldots,n_m,\,n_{m+1},\ldots, n_k$ is a set of
equations $n_i=q_i(n_{m+1},\ldots,\, n_{k})$ where $1\le i\le m$, then the expressions for $n_i$ can be substituted into $p=p'$ and one
trivially checks the equality of the two polynomials over
$n_{m+1},\ldots,\ n_{k}$ in the axiomatics of the rational field.
Recall that two polynomials are equal if and only if the coefficient
at monomials with the same degrees of variables are equal.

\subsection{Syntactical condition for decidability}
\label{typechecking:dec}

The simplest way to ensure decidability is to require that all equations in $D$
have the form $n=c$, where $c$ is a constant. This would in particular exclude the example $e_{H}$ from above. As we will see below, this requirement can be fulfilled by imposing
the syntactical condition for program expressions,
\textit{prohibiting pattern matching on variables other than function parameters
and bounded to them by other pattern matchings.}

It is easy to see
that any function body that satisfies the syntactic condition
may be encoded in the language defined by the \textit{refined grammar}
where the $\mathsf{let}$-construct in $e$ is replaced by
$\lete{\xarg}{b}{e_\mathit{nomatch}}$:

$$
\begin{array}{llll}
\basics & b & ::= &  c \ |\  x\,\mathsf{binop}\,y\
| \ \nil \ | \ \cons{z}{l} \ |\ f(z_1, \ldots, z_n)\\
\exprs & e & ::= &  \ b\\
& & &  |\ \lete{z}{b}{e_{\mathit{nomatch}}}\\
& & &  |\ \ifte{\xarg}{e_1}{e_2}\\
& &  & | \ \matchshort{l}{\hdarg}{\tlarg}{e_1}{e_2} \\
& & & | \ \letfune{f}{z_1, \ldots, z_n}{e_1}{e_2}\\
& &  & | \ \letexterne{f}{z_1, \ldots, z_n}{e_1} \\
\end{array}
$$

\noindent with
$$
\begin{array}{llll}
& e_{\mathit{nomatch}} & := & b  \\
& &    &  | \ \lete{z}{b}{e'_\mathit{nomatch}}\\
& &    &  | \ \ifte{x}{e'_\mathit{nomatch}}{e''_\mathit{nomatch}}\\
& & & | \ \letfune{f}{z_1, \ldots, z_n}{e}{e'_\mathit{nomatch}} \\
& & & | \ \letexterne{f}{z_1, \ldots, z_n}{e'_\mathit{nomatch}}\\
\end{array}
$$

\noindent The grammar is more restrictive than the syntactic condition. However,
any function body that satisfies the condition may be encoded in this grammar.
For instance, an expression
\vspace{0.1cm}

$
\lete{l'}{f_0(z)}
{\begin{array}[t]{l}
\match{l}{\hdarg}{\tlarg}
{f_1(l, \, l')}{f_2(l,\ l')}
\end{array}
}
$
\vspace{0.1cm}

\noindent and the expression
\vspace{0.1cm}

$
\match{l}{\hdarg}{\tlarg}
{\lete{l'}{f_0(z)}{f_1(l, \, l')}}
{\lete{l'}{f_0(z)}{f_2(l,\, l')}}
$
\vspace{0.1cm}

\noindent define the same map of lists.

For this reason we call the refined grammar the ``no-let-before-match'' grammar,
and roughly refer to the syntactic conditions as to the ``no-let-before-match'' condition.
The demo version of the type checker, accessible from \texttt{www.aha.cs.ru.nl},
uses the ``no-let-before-match'' grammar.

\begin{thm}
Let a program expression $e$ satisfy the refined grammar, and
let us check the judgement $\judgements{\ttt;\
x_1:\tau^o_1,\ldots,\,x_k:\tau^o_k}{e}{\tau}{\Sigma}$.
Then, at the end of the type-checking procedure
 one has to
check entailments of the form
\vspace{0.1cm}

$D\,\vdash\,p'=p,$
\vspace{0.1cm}

\noindent where $D$ is a set of equations of the form
$n-c=0$ for some $n\in\FVS{\tau^o_1\times \ldots \times \tau^o_k}$ and constant $c$
and $p$, $p'$ are polynomials in $\FVS{\tau^o_1\times \ldots \times \tau^o_k}$.
\end{thm}

\textit{Sketch of the proof}. Consider a path in the type checking
tree which ends up with some $D\vdash p'=p$ and let an equation
$q=0$ belongs to $D$. It means that in the path there is the
nil-branch of the pattern matching for some $l:\st{\tau}{q}$.

By induction on the length of the path, one can show that $q=n-c$
for some size variable $n\in\FVS{\tau_1\times \ldots \times
\tau_k}$ and some constant $c$. This uses the fact that follows from the
syntactic condition: the program variables
which are not free in a program expression and pattern-matched may be
introduced only by another pattern-matching, but not a let-binding.
The technical report  \cite{techrep} contains the full proof.

Of course, the syntactical condition of the theorem may be relaxed.
One may allow expressions with pattern-matching in a let-body,
assuming that functions that appear in let-bindings, like $f_0$,
give rise to solvable Diophantine equations. For instance, when
$p_0$ is a linear function, one of the variables is expressed via
the others and constants and substituted into $p_1=p$. Another
case when it is easy to check if there are natural
roots for $p_0=0$ or not (and find them if ``yes'') is
when $p_0$ is a 1-variable polynomial. We leave
relaxations of the condition for future work.

\section{Type Inference}
\label{inference}

Here we discuss type inference under the syntactical condition
defined in the previous section.
Since we consider shapely functions,
there is a way to reduce type inference to type-checking using
the well-known fact that a finite polynomial is defined by a finite
number of points. The procedure presented in this section
was sketched by us in \cite{ShvKvE07b}
and given in details and evaluated with a series of measurements in
\cite{vKShvE07}.

For each size dependency
from the output type of  a given function definition one
assumes that it is a polynomial and one guesses its degree. Then, to
obtain the coefficients of the polynomial of this degree, the function definition is evaluated (preferably in a
sand-box) as many times as the number of coefficients the polynomial
has. This finite number of input-output size pairs defines a system
of linear equations, where the unknowns are the coefficients of the
polynomial. When the sizes of the input data
satisfy some criteria known from  polynomial
interpolation theory \cite{Chui87,Lor92} (see the subsections below for more detail),
the system has a unique
solution. Input sizes that satisfy these criteria, which are
nontrivial for multivariate polynomials, can be determined
algorithmically.

In this way we find using interpolation theory
the interpolating polynomial for the size dependency.
If the size dependency is a polynomial function and the hypothesis about
its degree is correct, then it coincides with its
interpolating polynomial. To
check if this is the case, the interpolating polynomial is given to the type
checking procedure. If it passes, it is correct. Otherwise, one
repeats the procedure for a higher degree of the size dependency.
Starting with degree zero\footnote{On can also start with a higher degree. If the degree of the solution happens to be lower than the initial degree, the solution will still be found since the found coefficients will be zero at the right places.}, the method
iteratively  constructs the interpolating polynomials
until the correct polynomial is found. It  does not terminate when
\begin{enumerate}[(1)]
\item the function under consideration does not terminate on test data,
\item the function is non-shapely,
\item the function is shapely but the type-checker rejects it due to
the type-system's incompleteness (see section \ref{completeness}).
\end{enumerate}

The method infers polynomial size dependencies for a nontrivial class
of shapely functions.
For instance, standard type inference for the underlying type system
yields that the function $\cprod$ has the underlying type
$\st{\alpha}{}\times\st{\alpha}{}\longrightarrow
\st{\st{\alpha}{}}{}$. Adding size annotations with unknown output
polynomials gives $\cprod :
\st{\alpha}{n}\times\st{\alpha}{m}\longrightarrow
\st{\st{\alpha}{p_2}}{p_1}$. We assume $p_1$ is quadratic so we have
to compute the coefficients in its presentation:
$$p_1(n, m)=a_{0,0} + a_{0,1}n + a_{1,0} m + a_{1,1}nm  + a_{0,2}n^2 + a_{2,0}m^2 $$
Running the function $\cprod$ on six pairs of lists of length $0$,
$1$, $2$ yields:
$$\begin{array}{l@{\quad}l@{\quad}l@{\quad}l@{\quad}l@{\quad}l@{\quad}l}
n & m & l_1 & l_2 & \mathsf{cprod}(l_1, l_2) & p_1(n, m) & p_2(n, m)
\\ \hline
0 & 0 & [] & [] & [] & 0 & ? \\
1 & 0 & [0] & [] & [] & 0 & ? \\
0 & 1 & [] & [0] & [] & 0 & ? \\
1 & 1 & [0] & [1] & [[0, 1]] & 1 & 2 \\
2 & 1 & [0,1] & [2] & [[0,2],[1,2]] & 2 & 2 \\
1 & 2 & [0] & [1,2] & [[0,1],[0,2]] & 2 & 2
\end{array}$$
\noindent The first three rows of the table are examples
of \textit{incomplete measurements},
where the size of the inner list is unknown, because the outer list is empty.
The last three rows are \textit{complete measurements}.

The test table defines the following linear system for the outer output
list:
$$
\begin{array}{rl}
a_{0,0} & =0\\
a_{0,0} + a_{0,1}  + a_{0,2} & =0\\
a_{0,0} + a_{1,0}  + a_{2,0} & =0\\
a_{0,0} + a_{0,1} + a_{1,0} + a_{0,2}  + a_{1,1} + a_{2,0} & =1\\
a_{0,0} + 2a_{0,1} + a_{1,0} + 4 a_{0,2}  + 2a_{1,1} + a_{2,0} & =2\\
a_{0,0} + a_{0,1} + 2 a_{1,0} + a_{0,2}  + 2 a_{1,1} + 4 a_{2,0} & =2\\
\end{array}
$$
\noindent The unique solution is $a_{1,1}=1$ and the rest of
coefficients are zero. To verify whether the interpolation is indeed
the size polynomial, one checks if $\cprod :
\st{\alpha}{n}\times\st{\alpha}{m}\longrightarrow
\st{\st{\alpha}{2}}{n*m}$. This is the case, as was shown in section
\ref{typechecking:examples}.

As an alternative way of finding the coefficients, one could try to
solve directly the (recurrence) equations defined by entailments $D\vdash p=p'$
that arise during construction of the type-inference tree for a function definition.
As we will see in subsection \ref{motivation}, it amounts to solving
systems that are nonlinear in general. By combining testing with type checking we
bypass nonlinear systems \cite{vKShvE07}.

However, test-based inference has a drawback: it is not fully static.
The procedure has dynamic aspects, since it is done not only in the underlying logic
of the type system (i.e. Peano arithmetic), but it involves
executing the interpreter of the programming language. A consequence of it may be that
inference for function definitions with external calls is based on
the semantics of another language. When the size dependency of the external function is known, this can be avoided by
\begin{enumerate}[$\bullet$]
\item modifying the interpreter of our language in such a way, that
in the case of an  external call it creates a ``fake'' object of the right size
(the size of the result of ``this'' external call), or
\item leaving the interpreter in intact, and creating
for any external function from its sized type  a ``fake'' function body in our language with the same size dependency as the external function.
\end{enumerate}
\noindent From an engineering point of view, the advantage of the second approach is that a standard interpreter can be used directly. We discuss the mechanism of
generating ``fake'' functions in \ref{inhabit}.

Ideally, one  would like to remove all dynamic aspects from type inference.
In our current research  towards fully static inference
we consider a modification of the method where instead of the interpreter of
the programming language one uses an abstract interpreter
in the form of a term-rewriting system of which the rewriting rules will correspond to equations in Peano
arithmetic. For instance, $\prgs$ is interpreted as
$p(n) \rightarrow n+p(n-1)$ together with $p(0) \rightarrow 0$.
We have presented preliminary results
in the technical report \cite{ShvET08}.


\subsection{Motivation for test-based inference}
\label{motivation}

Consider, as an example of the complexity of systems generated by
conventional type inference, the  system for a function definition $\mathsf{nonlinear}$
with auxiliary functions:

\begin{center}
\begin{small}
\begin{tabular}{ll}
$\cop$:& $\st{\alpha}{n} \rightarrow \st{\alpha}{n}$\\
$\mathsf{copyfirst}$:& $\st{\alpha}{n_1} \times \st{\alpha}{n_2} \rightarrow
\st{\alpha}{n_1*n_2}$\\
$\mathsf{sqdiffaux}$:& $\st{\alpha}{n_1} \times \st{\alpha}{n_2} \rightarrow
\st{\alpha}{n_1^2+n_2^2-2*n_1*n_2}$\\
\end{tabular}
\end{small}
\end{center}

\noindent
where (in the sugared syntax\footnote{Recall, that in the sugared syntax we
use $f(g(z))$ for ``$\lete{z'}{g(z)}{f(z')}$'' and, moreover,
use
$[1\ldots c]$ for $c$-ary application of $\cons{-}{-}$ to $\nil$,
so that $[1\ldots 3]$ denotes $\cons{1}{\cons{2}{\cons{3}{\nil}}}$. We also use
the infix \texttt{++} for $\append$.})

$$
\begin{array}{l}
\letfun{\cop}{l}{\match{l}{\hdarg}{\tlarg}{\nil}{\cons{\hdarg}{\cop(\tlarg)}}}
{\letfun{\mathsf{copyfirst}}{l_1,\,l_2}{\match{l_2}{\hdarg}{\tlarg}{\nil}
{l_1\,\texttt{++}\,\mathsf{copyfirst}(l_1,\,\tlarg)}}
{
\letfun{\mathsf{sqdiffaux}}{l_1,\,l_2}
{
\begin{array}[t]{l}
\match{l_1}{\hdarg}{\tlarg}{\mathsf{copyfirst}(l_2,\,l_2)}{}\\
\quad \match{l_2}{\hdarg'}{\tlarg'}{\mathsf{copyfirst}(l_1,\,l_1)}{\mathsf{sqdiffaux}(\tlarg,\tlarg')}
\end{array}
}
{
\letfun{\mathsf{nonlinear}}{l_1,\,l_2}
{
\begin{array}[t]{l}
\match{l_1}{\hdarg}{\tlarg}{\mathsf{copyfirst}(\mathsf{copyfirst}(l_2,\,l_2),\;[1 \ldots 4])}{}\\
\begin{array}{l}
\match{l_2}{\hdarg'}{\tlarg'}{\mathsf{copyfirst}(\mathsf{copyfirst}(l_1,\,l_1),\;[1 \ldots 4])}{}\\
\mathsf{sqdiffaux}(\mathsf{nonlinear}(\tlarg,\,l_2) \, \texttt{++}\, l_1, \;
\mathsf{nonlinear}(l_1,\,\tlarg') \,\texttt{++} \,l_2)\,\\
\quad \texttt{++}\; \mathsf{copyfirst}(\mathsf{copyfirst}(l_1,\,l_2),\;[1 \ldots 17])
\end{array}
\end{array}
}
{\ldots}
}
}}\\
\end{array}
$$

The inference procedure ends up
with the following recurrence system:

$$
\begin{array}{lr}
\left\{
\begin{array}{lll}
p(0, n_2) & = &  4 n_2^2\\
p(n_1, 0) & = &  4 n_1^2 \\
p(n_1, n_2) & = & (p(n_1-1, n_2) +n_1- (p(n_1, n_2-1) + n_2))^2 + 17n_1n_2\\
\end{array}
\right. & (1)
\end{array}
$$

\noindent The  problem is \textit{to find} $p$, assuming, say, that it is quadratic.

A standard way  of solving this problem uses the method of unknown coefficients.
A polynomial to find, $p(n_1, n_2)$, is presented in the form
$a_{0,0}+a_{0,1}n_1+a_{1,0}n_2+a_{1,1}n_1n_2+a_{0,2}n_1^2+a_{2,0}n_2^2$ and substituted into  $(1)$.
Equating the corresponding coefficients of the polynomials from the left and right sides
of the equations from $(1)$ gives
$$
\left\{
\begin{array}{rll}
a_{0,0} & =& 0,\;\; a_{1,0} = 0,\;\; a_{2,0}  = 4,\;\; a_{0,1}  = 0,\;\; a_{0,2}  =4 \\
a_{0,2} & = & (a_{1,1}- 2a_{0,2}+1)^2 \\
a_{2,0} & = & (2a_{2,0}-a_{1, 1}-1)^2 \\
a_{1,1} & = & 2 (a_{1,1}- 2a_{0,2}+1)(2a_{2,0}-a_{1, 1}-1)+17 \\
a_{0,1} & = & 2((a_{1,0}-a_{0,1})+(a_{0,2}-a_{2,0}) )(a_{1,1}- 2a_{0,2}+1)\\
a_{1,0} & = & 2((a_{1,0}-a_{0,1})+(a_{0,2}-a_{2,0}) )(2a_{2,0}-a_{1, 1}-1)\\
a_{0, 0} & =& ((a_{1,0}-a_{0,1})+(a_{0,2}-a_{2,0}) )^2\\
\end{array}
\right.
$$


Substituting the coefficients $a_{0,0}=0,\; a_{1,0} = 0,\; a_{2,0}  = 4,\;a_{0,1}  = 0,\;a_{0,2}  =4$
in the remaining equations one obtains the non-linear system

$$
\left\{
\begin{array}{lll}
a_{1,1}^2 - 14  a_{1,1} + 45 & = & 0\\
2 a_{1,1}^2 - 27  a_{1,1} + 81& =  & 0 \\
\end{array}
\right.
$$
\noindent
The solution of this quadratic system can be found easily. It is  $a_{1, 1}=9$.

In general, non-linear systems may be hard to solve.
With the testing approach we
avoid solving \textit{nonlinear} systems w.r.t. polynomial coefficients $a_{ij}$.
Instead, we compute the coefficients solving
the \textit{linear} system that is generated
after testing.

\subsection{Interpolating a polynomial}

A hypothesis for a type is derived automatically by fitting a polynomial to the size
data, as it was shown in the example $\cprod$.
We are looking for the polynomial that best approaches the
data, i.e., the polynomial interpolation. The polynomial interpolation
exists and is unique under some conditions on the data, which are explored in
polynomial interpolation theory~\cite{Chui87,Lor92}.

For $1$-variable interpolation this condition is well-known.
A polynomial $p(z)$ of
degree $d$ with coefficients $a_1, \ldots, a_{d+1}$ can be written
as follows:
\begin{equation*}
a_1 \ +\ a_2 \, z\ + \ \ldots\  \ +\ a_{d+1}\,z^d  = \ p(z)
\end{equation*}
The values of the polynomial function in any pairwise different $d+1$ points determine
a system of linear equations w.r.t. the polynomial coefficients.
More specifically, given the set $\big(z_i, p(z_i)\big)$
 of pairs of numbers, where $1\leq i \leq d+1$, and coefficients
$a_1,\ \ldots \ , a_{d+1}$, the set of equations can be
represented in the following matrix form, where only the $a_i$ are
unknown:
$$
\left(
\begin{array}{@{}c@{\ }c@{\ }c@{\ }c@{\ }c@{}}
1 & z_1  &  \cdots & z_1^{d-1}  & z_1^d\\
1 & z_2  &  \cdots & z_2^{d-1}  & z_2^d\\
\vdots & \vdots  & \ddots & \vdots   & \vdots \\
1 & z_{d}  &  \cdots & z_{d}^{d-1}  & z_{d}^d\\
1 & z_{d+1}  &  \cdots & z_{d+1}^{d-1}  & z_{d+1}^d\\
\end{array}
\right) \left(
\begin{array}{@{}c@{}}
a_1 \\
a_2 \\
\vdots \\
a_d \\
a_{d+1}\\
\end{array}
\right) =  \left(
\begin{array}{@{}c@{}}
p(z_1) \\
p(z_2) \\
\vdots \\
p(z_d) \\
p(z_{d+1})\\
\end{array}
\right)
$$
\noindent The determinant of the left matrix, contains the
measurement points, is called a \textit{Vandermonde} determinant. For
pairwise different points $z_1,\ldots, \ z_{d+1}$ it is non-zero.
This means that, as long as the output size is measured for $d+1$
different input sizes, there exists a unique solution for the system
of equations and, thus, a unique interpolating polynomial.

The condition under which there exists a unique polynomial that
interpolates \emph{multivariate} data is not trivial. We formulate it in
the next subsection. Here we introduce the necessary definitions.

Recall  that a polynomial of degree $d$ and dimension $k$ (the number of variables)
has $N_d^k=\binom{d+k}{k}$ coefficients. Let a set of values $f_i$ of a real function $f$
be given. A set $W = \{\bar{w}_i: \,i=1,\ldots,\, N_d^k\}$ of points
in a real $k$-dimensional space forms the set of \textit{interpolation nodes} if
there is a unique polynomial
$p(\bar{z})=\Sigma_{0 \leq |j| \leq d}a_j \bar{z}^j$ with the total
degree $d$ with the property
 $p(\bar{w_i})=f_i$, where $1 \leq i \leq N_d^k$.
In this case one says that the polynomial $p$ interpolates the function
$f$ at the nodes $\bar{w}_i$.

The condition on $W$, which assures the existence and uniqueness of an interpolating polynomial,
is geometrical: it describes a configuration, called \textbf{NCA} \cite{Chui87},
in which the points from $W$ should be placed in $\R^k$. The multivariate
Vandermonde determinant computed from such points is
non-zero. Thus, the corresponding system of linear equations w.r.t.
the polynomial's coefficients has a unique solution.
In the following subsections we
show how to generate a collection of \textit{natural-valued} nodes  $\bar{w}_i$
in an \textbf{NCA} configuration.
A Vandermonde determinant is computed by the same formula in reals and naturals,
so the system of linear equations
based on natural nodes will have a unique (rational) solution.


\subsection{Measuring bivariate polynomials}
\label{SECTION:HYPOTHESIS:CONDITION}

For a two-dimensional polynomial of degree $d$, the condition on the
nodes that guarantees a unique polynomial interpolation is as
follows \cite{Chui87}:


\begin{defi} $N_d^2$ nodes forming a set $W \subset \R^2$ lie
in a \textit{2-dimensional NCA configuration} if there exist lines
$\gamma_1,\ldots,\gamma_{d+1}$ in the space $\R^2$, such that $d+1$ nodes of $W$ lie on
$\gamma_{d+1}$ and $d$ nodes of $W$ lie on  $\gamma_d \setminus
\gamma_{d+1}$, ..., and finally $1$ node of $W$ lies on  $\gamma_1 \setminus
(\gamma_2 \cup \ldots \cup \gamma_{d+1})$.
\end{defi}

An example of such a configuration for integers is given in figure \ref{Figure}a.

Nodes satisfying this condition can be found automatically:
if the output type of a given function definition
is $\st{\ldots\st{\alpha}{p_s}\ldots}{p_1}$, then
for the outermost-list size $p_1$ choose a triangle of nodes
on parallel lines, like in figure \ref{Figure}b.

An example of the two dimensional case is the $\cprod$ function above. As we have seen,
the procedure of reconstructing the size polynomial $p_1$ for the outer list is straightforward.
However, there is a problem for $p_2$. There are cases in which nodes
have no corresponding output size (the question-marks in the table that
refer to incomplete measurements). Measurements
for $p_2$ may be incomplete, because
the size of the inner lists can only be determined when there is at
least one such a list. Thus, the outer list may not be empty for complete measurements. As can
be seen in figure \ref{Figure}d, for $\cprod$ output's outer list is empty
when one of the two input lists is empty. In the next section,
we show that, despite this, it is always possible to find
enough measurements and give an upper bound on the number of natural nodes
that have to be searched.

\subsection{Handling incomplete measurements}
\label{SECTION:HYPOTHESIS:MEASUREMNTS}
\label{SECTION:HYPOTHESIS:INCOMPLETE}


In general, for $\st{ \ldots \st{\alpha}{p_s}
\ldots}{p_1}$ we will not find a value for $p_j$ at a node if one of
the outer polynomials, $p_1$ to $p_{j - 1}$, is zero at that node.
Thus, the nodes where $p_1$ to $p_{j - 1}$ are zero should be
excluded from the testing process. Here, we show that, despite this,
it is always possible to find enough nodes using finite search.

\begin{figure}[t]
\begin{tabular}{cccc}
\includegraphics[width=29mm]{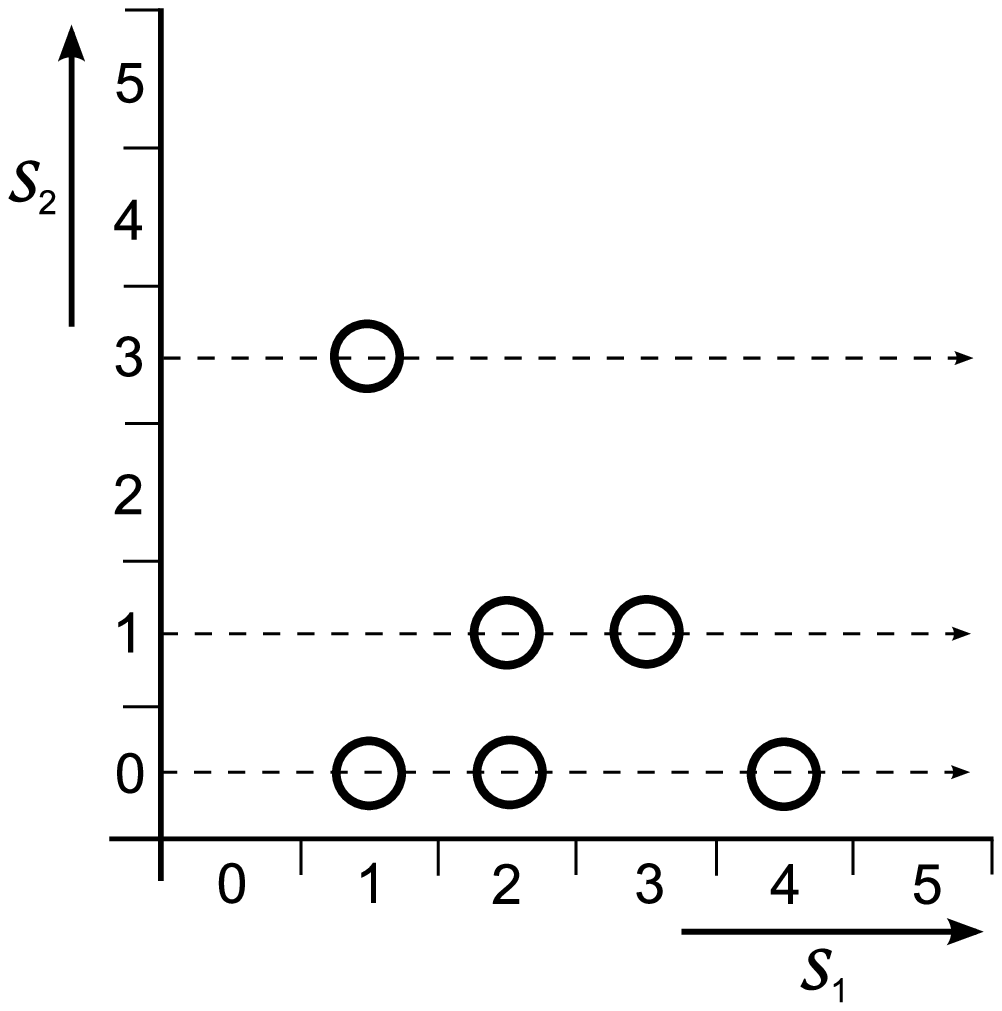}&
\includegraphics[width=29mm]{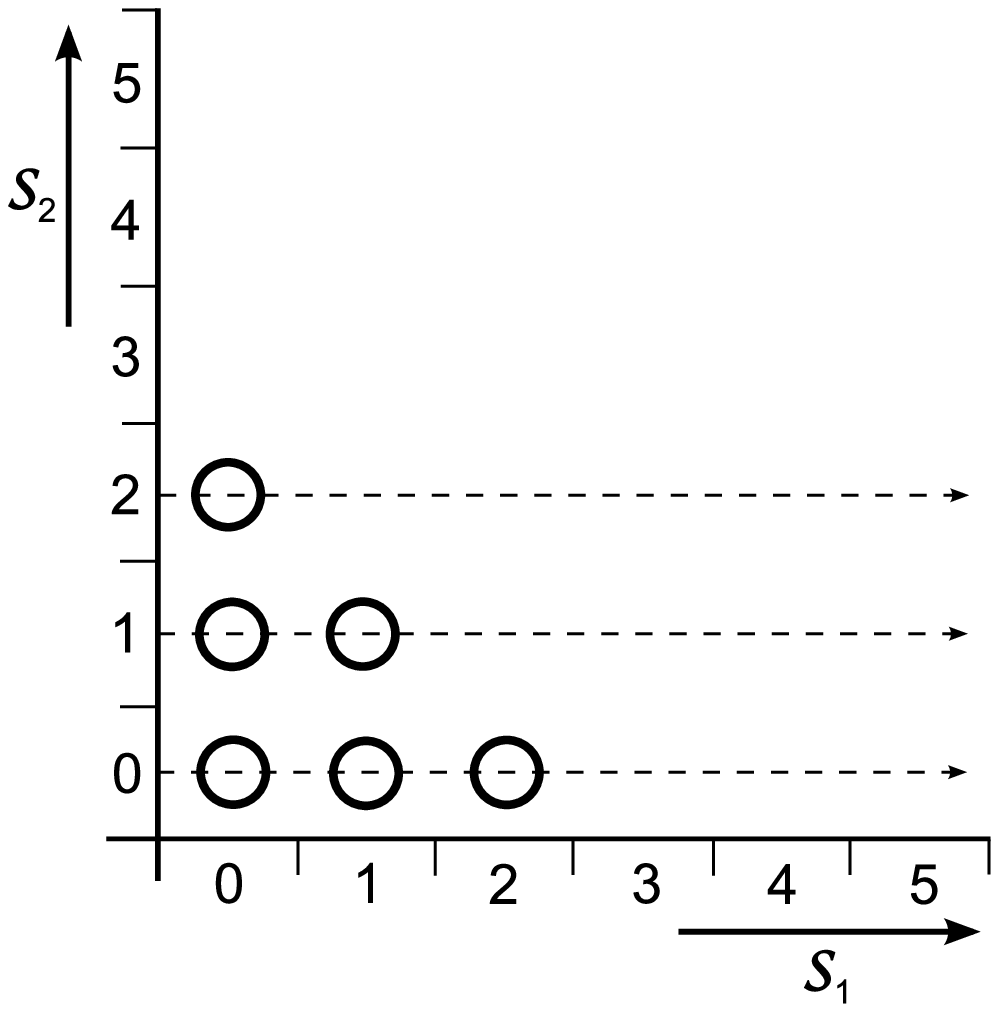}&
\includegraphics[width=29mm]{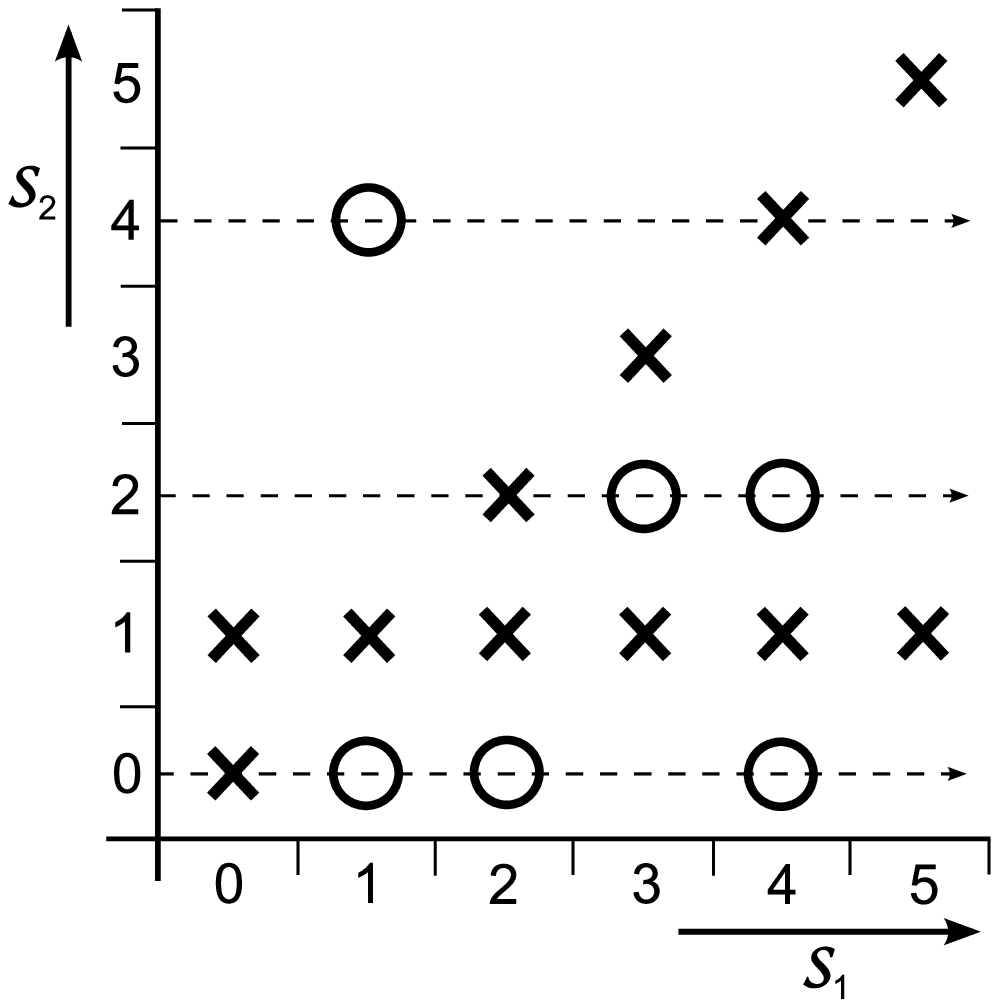}&
\includegraphics[width=29mm]{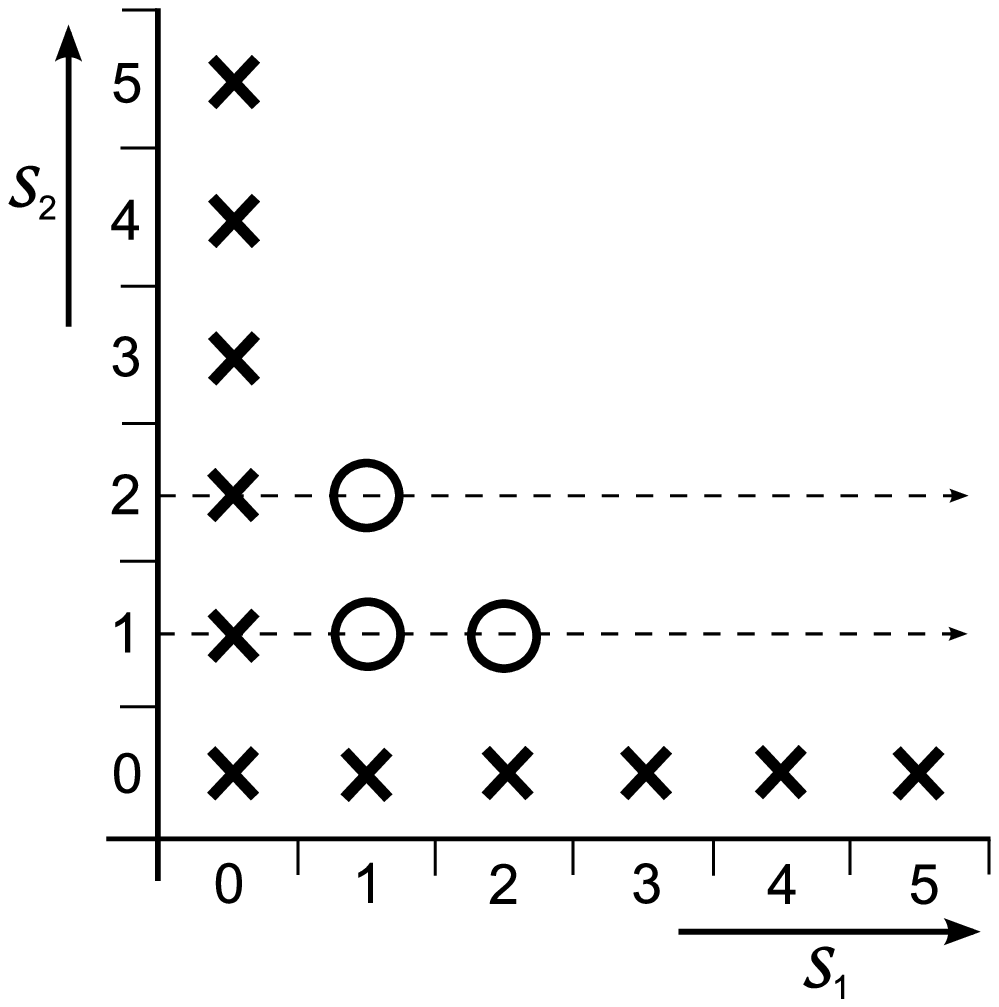}\\
(a) & (b) & (c) & (d)
\end{tabular}
\caption{(a) A node configuration that has a unique two-dimensional
polynomial interpolation (b) A more systematic node configuration
that has a unique two-dimensional polynomial interpolation (c)
Incomplete measurements complicate finding a node configuration (d)
Incomplete measurements for the pairs in the output of $\cprod$.}
\label{Figure}
\end{figure}

First,  nested output lists of which the size of the
outer list is the constant zero, e.g. $\st{\st{\alpha'}{p_2}}{0}$, need special
treatment.
If a type-checker rejects annotations for $p_1 \equiv 0$ and \textit{arbitrary} $p_2$
then the outer polynomial $p_1$ is not a constant zero. (Recall the definition
of $D\vdash \tau = \tau'$.)

Now, let  the outer polynomial $p_1(x, y)$ be not a constant zero.
Then there is a finite number of lines $y = i$, which
we will call \textit{root lines}, where $p_1(x,\,i) = 0$.

\begin{lemma}
\label{LEMMA:zerolines}
A polynomial $p_1(x,\,y)$ of degree $d$ that
is not constant $0$ has at most $d$ root lines $y=i$, such that
$p_1(x,\,i) = 0$ for all $x$.
\end{lemma}
\proof
Suppose there are more than $d$ root lines. Then, it is easy to pick
$1, \ldots, d+1$ nodes on $d+1$  root lines. They trivially are in  \textbf{NCA} configuration.
With these nodes, at which $p_1(x, y) = 0$, the system of linear equations for the
coefficients of $p_1$ will have the zero-solution, that is, all the
coefficients of $p_1$ will be zeros. This contradicts the assumption
that $p_1$ is not constant 0.
\qed

Using the lemma, we can bound the number
of parallel lines $y=i$ and nodes on them that have to be searched.
Essentially, we are to find a triangle configuration of nodes, like on
figure \ref{Figure}b, skipping all crosses, see \ref{Figure}c.

\begin{lemma}
When looking for nodes for a polynomial $p_2(x,y)$ that
determine a unique polynomial interpolation at places where another
polynomial $p_1(x,y) \ne 0$, it is sufficient to search the lines $y =
0, \ldots, y = d_1 + d_2$ in the square $[0, \ldots, d_1 +
d_2]\times[0, \ldots, d_1 + d_2]$.
\end{lemma}
\proof
For the configuration it is sufficient to have $d_2 + 1$ lines $y=i$ with
at least $d_2 + 1$ points where $p_1(x, y) \not= 0$. Due to lemma
\ref{LEMMA:zerolines} there are at most $d_1$ lines $y=i$ such that
$p_1(x,\,i)=0$, so at least $d_2+1$ are not root lines for $p_1$. The
polynomial $p_1(x,\,j)$, with $y=j$ not a root line, has at most
degree $d_1$, thus $y=j$ contains at most $d_1$ nodes $(x,\,j)$,
such that $p_1(x,\,j)=0$. Otherwise, it would have been constant zero,
and thus a root line. Hence, this leaves at least $d_2 + 1$ points
on these lines for which $p_1$ is not zero.
\qed

This straightforwardly generalizes to all nested types
$\st{ \ldots \st{\alpha}{p_s}\ldots}{p_1}$ with
polynomials in two variables.
If we want to derive the coefficients of $p_i$,
searching the square of input values $[0, \ldots, \Sigma_{j=1}^{i}
d_j]\times[0, \ldots, \Sigma_{j=1}^{i} d_j]$ suffices, where $d_j$
is the degree of $p_j$. Each $p_j$ has at most $d_j$ root lines, so
there are at most $\Sigma_{j=1}^{i - 1} d_j$ root lines for $p_1,\ldots, p_{i-1}$. Also, each
of the $p_j$ can have at most $d_j$ zeros on a non root line. Hence,
since the length
of the search interval for $p_i$ is $\Sigma_{j=1}^{i} d_j + 1$, there are always $d_i
+ 1$ values known.

Eventually, it is enough to search in $[0, \ldots, \Sigma_{j=1}^{s}
d_j]\times[0, \ldots, \Sigma_{j=1}^{s} d_j]$.

For $\cprod$ there are two size expressions to derive, $p_1$ for the
outer list and $p_2$ for the inner lists. Deriving that $p_1(n_1,
n_2) = n_1 * n_2$ is no problem. Because $p_1$ has roots for $n_1 =
0$ and for $n_2 = 0$, these nodes should be skipped when measuring
$p_2$ (see figure \ref{Figure}d).

\subsection{Generalizing to k-dimensional polynomials}
\label{SECTION:HYPOTHESIS:COMPLEXITY}

 The generalization of the
condition on nodes for a unique polynomial interpolation to
polynomials in $k$ variables, is a straightforward inductive
generalization of the two-dimensional case. In a hyperspace there
have to be hyperplanes, on each of which nodes lie that satisfy the
condition in the $k-1$ dimensional case. A hyperplane $K_j^k$ may be
viewed as a set in which test points for a polynomial of $k-1$
variables of the degree $j$ lie. There must be
$N^{k-1}_j=N_j^k-N_{j-1}^k$ such points. The condition on the nodes
is defined by:
\begin{defi}
\label{DEF:NCAk}
The \textit{NCA configuration for $k$ variables ($k$-dimensional
space)} is defined inductively on $k$ \cite{Chui87}. Let
$\{\bar{z}_1,\ldots,\,\bar{z}_{N_d^k}\}$  be a set of distinct points in
$\mathcal{R}^k$ such that there exist $d+1$ hyperplanes $K^k_j$, $0
\leq j \leq  d$ with

$
\begin{array}{l}
\bar{z}_{N_{d-1}^k+1},\ldots,\,\bar{z}_{N_d^k} \;\in \;K_d^k\\
\bar{z}_{N_{j-1}^k+1},\ldots,\,\bar{z}_{N_j^k} \;\in
\;K^k_j\setminus\{K^k_{j+1}\cup\ldots\cup K^k_{d}\},
\mbox{for}\;\;0\leq j \leq d-1
\end{array}
$

\noindent and each of set of points
$\bar{z}_{N_{j-1}^k+1},\ldots,\,\bar{z}_{N_j^k}$, $0\leq j \leq d$, considered
as points in $\mathcal{R}^{k-1}$ satisfies \textbf{NCA} in
$\mathcal{R}^{k-1}$.
\end{defi}

For instance, given $d=2$ and $k=3$ (i.e. interpolating by polynomials
of $3$ variables of degree $2$), the following collection of $N_2^3=\binom{2+3}{3}=10$ nodes,
placed on parallel planes in $\R^3$,
satisfies an  \textbf{NCA} configuration:
\begin{enumerate}[(1)]
\item on the plane $x=0$ take the ``triangle'' of $N_2^2=6$ points $(0,\,0,\,0)$,
$(0,\,0,\,1)$, $(0,\,0,\,2)$, $(0,\,1,\,0)$, $(0,\,1,\,1)$,  $(0,\,2,\,0)$,
\item on the plane $x=1$ take the ``triangle'' of $N_1^2=3$ points
$(1,\,1,\,0)$, $(1,\,0,\,1)$,  $(1,\,1,\,1)$,
\item on the plane $x=2$ take the point $(2,\,0,\,0)$.
\end{enumerate}
\noindent Here the nodes on each of the planes lie in the $2$-dimensional  \textbf{NCA} configurations
constructed for degrees $2$, $1$ and $0$ respectively.

Similarly to lines in a square in the two-dimensional case,
parallel hyperplanes in $\R^k$  have to be searched
while generating hypothesis for a nested type. Using a
reasoning similar to the two-dimensional case one can show that it
is always sufficient to search a hypercube with sides $[0, \ldots,
\Sigma_{j = 1}^{s} d_j]$.

\subsection{Automatically inferring size-aware types: the procedure}
\label{SECTION:COMBINING}


The type checking procedure and the size
hypothesis generation can be combined to create an inference procedure.
The procedure starts with assuming a fixed degree. The assumptions is that this degree is the maximum degree of all polynomials in the type. If checking
rejects the  hypothesis generated for this degree, the degree is increased
and the test-check cycle is repeated.
The procedure is
semi-algorithmic: it terminates only when the function is well-typable.

Recently, we have developed a demonstrator for the inference procedure
described in \cite{vKShvE07}. It is accessible on \texttt{www.aha.cs.ru.nl}.

For any shapely program, the underlying type (the type
without size annotations) can be derived by a standard type
inference algorithm \cite{milner}. After straightforwardly
annotating input sizes with size variables and output sizes with
size expression variables, we have for example
$$
\cprod : \st{\alpha}{n_1} \times
\st{\alpha}{n_2} \rightarrow
\st{\st{\alpha}{p_2(n_1,n_2)}}{p_1(n_1,n_2)}
$$

To derive the size expressions on the right hand side we use the
following procedure. First, the maximum degree of the occurring size
expressions is assumed, starting with zero. Then, a hypothesis is
generated for each size expression, from $p_1$ to $p_s$.
After hypotheses have been obtained for all size expressions they are added to the type
and this hypothesis type is checked using the type checking
algorithm. If it is accepted, the type is returned. If not, the
procedure is repeated for a higher degree $d$.

The schema below shows the procedure in pseudo-code. The \textit{TryIncreasingDegrees}
function generates (by \textit{GetSizeAwareType}) and checks (by \textit{CheckSizeAwareType})
hypotheses. A size expression is derived by selecting a node configuration (\textit{GetNodeConf}),
running the tests for
these nodes (\textit{RunTests}), and deriving the size polynomial from the test
results (\textit{DerivePolynomial}).

\vspace{0.3cm}

\framebox{
\begin{tabular}{l}
\textsf{Function:} \textsc{TryIncreasingDegrees}\\
\textsf{Input:} a degree $d$, a function definition \texttt{f}\\
\textsf{Output:} the size-aware type of that function\\
\\
\begin{tabular}{@{}l@{\ }l@{\ }l@{\ }l@{}}
\multicolumn{4}{@{}l@{}}{\textsc{TryIncreasingDegrees}(\textit{d}, \texttt{f}) =}\\
\quad\ & let & \textit{type} & \textsf{=}  \textsc{InferUnderlyingType}(\texttt{f})\\
&     & \textit{atype} & \textsf{=}  \textsc{AnnotateWithSizeVariables}(\textit{type})\\
&     & \textit{vs} & \textsf{=}  \textsc{GetOutputSizeVariables}(\textit{atype})\\
&     & \textit{stype} & \textsf{=}  \textsc{GetSizeAwareType}(\textit{d}, \texttt{f}, \textit{atype}, \textit{vs}, [\ ])\\
& in  & \multicolumn{2}{@{}l@{}}{if (\textsc{CheckSizeAwareType}(\textit{stype}, \texttt{f})) then \textit{stype}}  \\
& & \multicolumn{2}{@{}l@{}}{else \textsc{TryIncreasingDegrees}(\textit{d+1}, \texttt{f})}\\
\end{tabular}\\
\\
\\
\begin{tabular}{ll}
\textsf{Function:} &\textsc{GetSizeAwareType}\\
\textsf{Input:} & a degree $d$, \\
& a function definition \texttt{f},\\
& its annotated type,\\
& a list of unknown size annotations, \\
& and the polynomials already derived\\
\textsf{Output:} & the size-aware type \\
& of that function if the degree is high enough\\
\end{tabular}
\\
\begin{tabular}{@{}l@{}}
\textsc{GetSizeAwareType}(\textit{d}, \texttt{f}, \textit{atype}, [\ ], \textit{ps}) =\\
\quad\ \textsc{AnnotateWithSizeExpressions}(\textit{atype},
\textit{ps})    // The End\\
\textsc{GetSizeAwareType}(\textit{d}, \texttt{f}, \textit{atype}, \textit{v}:\textit{vs}, \textit{ps}) =\\
\quad
\begin{tabular}{@{}l@{\ }l@{\ }l@{}}
let & \textit{nodes} & = \textsc{GetNodeConf(\textit{d}, \textit{atype}, \textit{ps})} \\
    & \textit{results} & = \textsc{RunTests}(\texttt{f}, \textit{nodes})\\
    & \textit{p} & = \textsc{DerivePolynomial}(\textit{d}, \textit{v}, \textit{atype}, \textit{nodes}, \textit{results})\\
in  & \multicolumn{2}{@{}l@{}}{\textsc{GetSizeAwareType}(\textit{d}, \texttt{f}, \textit{atype}, \textit{vs}, \textit{p}:\textit{ps})}\\
\end{tabular}
\end{tabular}
\end{tabular}
}
\vspace{0.3cm}

\noindent If a type is rejected, this can mean two things. First, the assumed
degree was too low and one of the size expressions has a higher
degree. That is why the procedure continues for a higher degree.
Another possibility is that one of the size expressions is not a
polynomial (the function definition is not shapely) or that the type
cannot be checked due to incompleteness of the type system. In that case the procedure
will not terminate.
If the function is well-typable, the
procedure will eventually find the correct size-aware type and terminate.

A collection of examples -- function definitions together with size measurements --
is presented in \cite{vKShvE07}.

\subsection{Complexity of hypotheses-generating phase}
\label{testcompl}

Given a function definition,
its underlying first-order type and a maximal degree of hypothetical polynomials,
the complexity of its hypothesis-generating phase depends on three parameters:
\begin{enumerate}[$\bullet$]
\item the nestedness $s \geq 0$ of the output type which may be either
$\st{\ldots \st{\intt}{p_s}\ldots}{p_1}$ or $\st{\ldots \st{\alpha}{p_s}\ldots}{p_1}$,
\item  the fixed maximal degree $d$ of the polynomials $p_1,\ldots,\,p_s$,
\item  the number of size variables $k$ defined by the input type of the
function.
\end{enumerate}

To generate hypothesis for $p_1(n_1,\,\ldots,\,n_k)$ one
\begin{enumerate}[(1)]
\item generates  $N^{k}_{d}=\binom{k+d}{k}$ natural-valued nodes inductively on $k$;
it is done by the definition \ref{DEF:NCAk} of \textbf{NCA} configuration for the $k$-variable case
(note that for $k=1$ it is just the $1$-dimensional nodes $0,\ldots,\,d$).
\item generates a collection of $N^{k}_{d}$ concrete inputs with the sizes, defined by the nodes,
\item evaluates the function body $N^{k}_{d}=\binom{k+d}{k}$ times on these inputs,
\item solves the system of $N^{k}_{d}$ linear equations to obtain  $N^{k}_{d}$ coefficients
for $p_1$.
\end{enumerate}

Generating hypotheses for a $p_j$, $j>1$, is similar. However, generating the collection
$N^{k}_{d}=\binom{k+d}{k}$ nodes is more complicated, since nodes sending some $p_{j'}$, $j'<j$, to zero
are excluded. In the worst case, to find correct nodes, one needs to
evaluate a $k$ dimensional cube with side $[0,\ldots,\,jd]$, that is
to evaluate (to check if it has a zero value)
$j-1$ polynomials in at most $(jd+1)^k$ nodes.

Thus, for each $1 \leq j \leq s$ the complexity is bounded by
$c_{\textit{eval}\;  p_1,\ldots, p_{j-1}} + c_{\textit{eval}\; p_j} + c_\textit{gauss}$,
where
\begin{enumerate}[$\bullet$]
\item $c_{\textit{eval} \; p_1,\ldots, p_{j-1}}=(j-1) \cdot (jd+1)^k$ evaluations of polynomials,
\item $c_{\textit{eval}\;  p_{j}}=N^{k}_{d}=\binom{k+d}{k}$ evaluations of the function definition,
\item $c_{\textit{gauss}}=O(N^{k\;2}_{d})$ is the complexity of Gaussian elimination.
\end{enumerate}

If the results of
evaluations of polynomials on the $j$-th step are memoised, then altogether for $j=1,\ldots\,s$
one needs at most $(s-1)\cdot (sd+1)^k$ evaluations of polynomials.
Thus, the complexity of the hypotheses-generating phase for  all $j=1,\ldots\,s$ together
is $(s-1)\cdot (sd+1)^k + s \cdot \binom{k+d}{k} + s \cdot O(\binom{k+d}{k}^2)$.

\subsection{Inhabitants for the types of external functions}
\label{inhabit}
Let $\f_{\mathsf{ext}}$ be an external function. Since the function is external, its code is not present in our language. However, its first-order type may be available. We have to trust this type since we cannot check it.

For inference of types of other functions that somewhere call $\f_{\mathsf{ext}}$, our testing procedure requires the possibility to evaluate within our language the code of the external function.
Such code can be made available in our language by constructing an inhabitant of the type of $\f_{\mathsf{ext}}$.

For our demonstrator, an alternative solution would be to create an actual external call for each occurrence of an external function. This may require more implementation effort within the demonstrator. The type inference procedure might take more time because the external function may require more time to execute than the generated inhabitants of the type.
Therefore, we prefer to work with inhabitants (which yields the same size dependencies as using external functions directly).
For reasons of modularity it might even be worthwhile to also create inhabitants of internal functions (e.g. in the case of using an interface to a huge, time intensive library).

Below, we show how to construct in our language a function $\f$ which is an inhabitant of a given type of an external function. It is not necessary to demand that $\f$ and the external function are equal as set-theoretic maps. They must have the same size dependency, i.e. the same type.

Let $\f_{\mathsf{ext}}$ have the type $\st{\alpha}{n} \rightarrow \st{\alpha}{p(n)}$. We define the body of $\f$ by the following program expression:

$$
\match{l}{\hdarg}{\tlarg}{\nil}
{
\mathsf{gen}\Big(\hdarg,\, \mathsf{p}(p) (\mathsf{length}(l))\Big)
}
$$

\noindent Now we explain the subexpressions in the nil- and cons-branches.
In the nil-branch the expression returns the empty list. This is the only
choice, due to the following ``folklore'' property (which to our knowledge was not published earlier).

\begin{lemma}
Any \textit{total} polymorphic function
$g\,:\,\st{\alpha}{}\rightarrow \st{\alpha}{}$ maps the empty list
to the empty list.
\end{lemma}
\proof
We prove this property using the ``free'' theorem $
\mathsf{map}({\mathsf{a}}) \circ g _{\alpha} = g_{\alpha'} \circ \mathsf{map}({\mathsf{a}})$
from \cite{Wad89}, which holds for all $\mathsf{a}\,:\,\alpha\rightarrow \alpha'$. Here $\mathsf{map}\,:\,(\alpha \rightarrow \alpha') \rightarrow
\st{\alpha}{}  \rightarrow \st{\alpha'}{}$ lifts $\mathsf{a}$ to lists,
and  $g_{\alpha}$ denotes the instantiation of $g$ with
type $\alpha$. Suppose the opposite:
$g_\alpha$ sends $\nil$ to $[\hdarg \ldots \mathit{stop}]$, and
$g_{\alpha'}$ sends $\nil$ to $[\hdarg' \ldots \mathit{stop}']$. Then
$\mathsf{map}({\mathsf{a}}) \circ g_\alpha$ sends $\nil$ to
$[\mathsf{a}(\hdarg)\ldots \mathsf{a}(\mathit{stop})]$
and  $g_{\alpha'} \circ \mathsf{map}({\mathsf{a})}$ sends $\nil$ to $[\hdarg'\ldots \mathit{stop}']$. It is not the case that for all $\mathsf{a}$ one has $\mathsf{a}(\hdarg)=\hdarg'$.\qed

\noindent It is a routine exercise to extend this ``property for free'' to nested lists.

In the cons-branch we use a straightforwardly defined function $\mathsf{gen}(z,\, x) : \alpha \times \intt \rightarrow \st{\alpha}{}$ that
outputs a list of $z$-s of length $x$ if $x$ is non-negative and does not terminate otherwise.
We also  use a function generator $\mathsf{p}$, that given a polynomial $p$,
generate a function definition $\mathsf{p}(p):\intt \rightarrow
\intt$ such that  $\mathsf{p}(p)(n)=p(n)$. It is easy to see
that for any non-empty list $l$ of length $n$
the composition $\mathsf{gen}(\hd,\,\mathsf{p}(p)(\mathsf{length}(l)))$ terminates
if $\f_{\mathsf{ext}}$ terminates. It follows from the fact
that if $\f_{\mathsf{ext}}$ terminates on $l$ then $p(n) \geq 0$, since $p(n)$ is the length of the corresponding output.

\section{Conclusion and Further Work}
\label{conclusion}

We have presented a natural syntactic restriction such that type
checking of a size-aware type system for first-order shapely
functions is decidable for polynomial size expressions without any
limitations on the degree of the polynomials.

A non-standard, practical method to infer types is introduced. It
uses run-time results to generate a set of equations. These
equations are linear and hence automatically solvable.
The method terminates on a non-trivial class
of shapely functions.

\subsection{Further work}

The system is defined for polymorphic lists. Recently, it has been shown~\cite{TFP08} how to extend the system to ordinary inductive types (no nested inductive definitions).

An obvious limitation of our approach is that we consider only
shapely functions. In practice, one is often interested to obtain
upper bounds on space complexity for non-shapely functions. A simple
example, where for a non-shapely function an upper bound would be
useful, is the function to $\mathsf{insert}$ an element in a list,
provided the list does not contain the element. At
present we
have been studying checking and inference of size annotations
in the form of collections of piecewise polynomials that represent
at least all possible size dependencies. For instance,
$\mathsf{insert}$ is annotated with $\{p(n)=n+i\}_{0 \leq i \leq 1}$,
and $\mathsf{delete}$ is annotated with $\{p(n)=n -^{\!\!\!\cdot} i\}_{0 \leq i \leq 1}$.
Such collections may be potentially infinite, like in the case
of recursive application of insert with $\{p(n, \,m)=n+i\}_{0 \leq i \leq m}$.
Here, involvement of real arithmetic is inevitable in type checking.
As for inference, when one  is interested in
strict (``principal type'') and polynomial lower and upper bounds,
$p_{\min}$ and $p_{\max}$ respectively,
it is possible to extend our testing procedure to
obtain them. Then, one checks the hypothesis in the form
$\{p_{\min}+i\}_{0 \leq i \leq (p_{\max} - p_{\min})}$.

We plan to allow both unsized integers and adding non-trivial sizes to integers.
The size of a non-negative sized integer
is taken to be its value. This allows to
type such functions as $\textsf{init}:\intt^n \rightarrow \st{\intt}{n}$,
which on the integer $n$ outputs the list of $1$ of length $n$.
With sized integers one can type such function definitions
without introducing dependent types.
Hence, the decision how to add sizes to integers is connected to the problem of using sized and non-sized types
within the same system. We leave it for future work based e.g. on
\cite{Vas03} and \cite{Jay97}.

Addition of other data structures and extension to non-shapely
functions will open the possibility to use the system for an actual
programming language.

Application of the methodology to estimate stack and time complexity
is considered as a topic for future projects.

\section*{Acknowledgments}

The authors would like to thank Alejandro Tamalet and the anonymous reviewers for their observations and valuable suggestions for improvement. We thank the students of Radboud University Nijmegen, --
Willem Peters, Bob Klaase, Elroy Jumpertz, Jeroen Claassens, Martin van de Goor and Ruben Muijrers --
without whom implementation of the on-line demonstrator would have been impossible.


\section*{Appendix: auxiliary lemmata for soundness proof}

\begin{lemma}[A program value's footprint is in the heap]
\label{LEMMA:FootprInHeap}$\Reg{h}{v} \subseteq \dom{h}$.
\end{lemma}
\begin{proof}
The lemma is proved by induction on the size of the (domain of the) heap $h$.
\begin{description}
\item[$\dom{h}=\emptyset$] Then no $\ell\in\dom{h}$ exists and
$\Reg{h}{v} = \emptyset$.
\item[$\dom{h}\ne\emptyset$]
\begin{description}
\item[$v=c$ or $v=\nul$] Then $\Reg{h}{v} = \emptyset$, which is trivially a subset of $\dom{h}$.
\item[$v=\ell$ and $\dom{h}=(\dom{h}\setminus\{\ell\})\cup\{\ell\}$]
From the definition of $\Reg{}{}$ we get  $\Reg{h}{\ell} = \{\ell\} \cup \Reg{\rstrloc{h}{\ell}}{h.l.\hd} \cup \Reg{\rstrloc{h}{\ell}}{h.l.\tl}$.
Applying the induction hypotheses we derive that
$\Reg{\rstrloc{h}{\ell}}{h.\ell.\hd} \subseteq \dom{\rstrloc{h}{\ell}}$ and
$\Reg{\rstrloc{h}{\ell}}{h.\ell.\tl} \subseteq \dom{\rstrloc{h}{\ell}}$. Hence, $\Reg{h}{l} \subseteq \dom{h}$.
\end{description}
\end{description}
\end{proof}

\begin{lemma}[Extending a heap does not change the footprints of program values]
\label{LEMMA:FootprExtHeap}If $\ell \notin \dom{h}$ and $h'=h[\ell.\hd:=v_\hd,\
\ell.\tl:=v_\tl]$ for some $v_\hd,\ v_\tl$ then for any $v \not=
\ell$ one has $\Reg{h}{v}=\Reg{h'}{v}$.
\end{lemma}
\begin{proof} The lemma is proved by induction on the size of the (domain of the) heap $h$.
\begin{description}
 \item[$\dom{h}=\emptyset$] Since $h' = [ \ell.\hd \:= v_\hd, \ell.\tl := v_\tl ]$ and $v \not= \ell$ we have 
$v \not\in \{\ell\}=\dom{h'}$. Therefore, $\Reg{h}{v} = \emptyset = \Reg{h'}{v}$.
 \item[$\dom{h}\ne\emptyset$] We proceed by case distinction on $v$.
\begin{description}
\item[$v=c$ or $v=\nul$] Then, $\Reg{h}{v} = \emptyset = \Reg{h'}{v}$.
\item[$v=\ell'$]
If $\ell'\notin\dom{h}$, then due to $\ell'\ne\ell$ we have $\ell'\notin\dom{h}$ as well
and $\Reg{h}{v} = \emptyset = \Reg{h'}{v}$.

Let  $\ell'\in\dom{h}$.
From the definition of $\Reg{}{}$ we get
$$\Reg{h}{\ell'} =
\{\ell'\}\ \cup\ \Reg{\rstrloc{h}{\ell'}}{h.\ell'.\hd}\ \cup\
\Reg{\rstrloc{h}{\ell'}}{h.\ell'.\tl}.$$

Due to
$h'(\ell')=h(\ell')$ and
$$\rstrloc{h'}{\ell'}=\rstrloc{h}{\ell'}[\ell.\hd:=v_\hd, \, \ell.\tl:=v_\tl],$$
\noindent
and the induction assumption
 one has
$$
\begin{array}{l}
\Reg{\rstrloc{h}{\ell'}}{h.\ell'.\hd}=\Reg{\rstrloc{h'}{\ell'}}{h'.\ell'.\hd}\\
\Reg{\rstrloc{h}{\ell'}}{h.\ell'.\tl}=\Reg{\rstrloc{h'}{\ell'}}{h'.\ell'.\tl}\\
\end{array}
$$
\noindent
So,
$$
\begin{array}{l}
\Reg{h'}{\ell'} =\\
=\{\ell'\}\ \cup\ \Reg{\rstrloc{h'}{\ell'}}{h'.\ell'.\hd}\ \cup\
\Reg{\rstrloc{h'}{\ell'}}{h'.\ell'.\tl}=\\
=\{\ell'\}\,\cup\,\Reg{\rstrloc{h}{\ell'}}{h.\ell'.\hd}\ \cup\
\Reg{\rstrloc{h}{\ell'}}{h.\ell'.\tl}=\\
=\Reg{h}{\ell'}.\\
\end{array}
$$

\end{description}
\end{description}

\end{proof}

\begin{lemma}[Extending heaps preserves model relations]
\label{LEMMA:ModelExtHeap}\ \\
For all heaps $h$ and $h'$, if $h'|_{\dom{h}}=h$ then $\mdl{v}{h}{\typesvg}{w}$ implies $\mdl{v}{h'}{\typesvg{}}{w}$.
\end{lemma}
\begin{proof}\ \\
The lemma is proved by induction on the structure of $\typesvg$.
\begin{description}
\item[$\typesvg = \intt$] In this case, $v$ is a constant $c$ and $w = c$, hence $\mdl{v}{h'}{\typesvg}{w}$ by the definition.
\item[$\typesvg = \st{\typesvg{}'}{n^\bullet}$] We proceed by induction on $n^\bullet$.
\begin{description}
\item[$n^\bullet = 0$] In this case, $v = \nul$ and $w = \mbox{\texttt{[]}}$, hence $\mdl{v}{h'}{\typesvg}{w}$ by the definition.
\item[$n^\bullet = m^\bullet + 1$]  By the definition $v$ is a location $\ell$ and $\mdl{\ell}{h}{\st{\typesvg{}'}{m^\bullet + 1}}{w_\hd :: w_\tl}$ for some $w_\hd$ and $w_\tl$ such that
$$
\begin{array}{l}
\ell \in\dom{h},\\
\mdl{h.\ell.\hd}{\rstrloc{h}{\ell}}{\typesvg{}'}{w_\hd},\\
\mdl{h.\ell.\tl}{\rstrloc{h}{\ell}}{\st{\typesvg{}'}{m^\bullet}}{w_\tl}
\end{array}
$$
\noindent
We want to apply the induction assumption, with heaps $h|_{\dom{h}\setminus\{\ell\}}$,
$h'|_{\dom{h'}\setminus\{\ell\}}$ (as ``$h$'' and ``$h'$'' respectively).
The condition of the lemma is satisfied because
$$
\begin{array}{l}
\restrict{\rstrloc{h'}{\ell}}{\dom{\rstrloc{h}{\ell}}}\\
=\restrict{\rstrloc{h'}{\ell}}{\dom{h}\setminus\{\ell\}}\\
=\restrict{h'}{\dom{h} \setminus\{\ell\}}=\rstrloc{h}{\ell}\\
\end{array}
$$
\noindent Thus, we apply the induction assumption and with $h.\ell = h'.\ell$ obtain
$$
\begin{array}{l}
\ell \in\dom{h'},\\
\mdl{h'.\ell.\hd}{\rstrloc{h'}{\ell}}{\typesvg{}'}{w_\hd},\\
\mdl{h'.\ell.\tl}{\rstrloc{h'}{\ell}}{\st{\typesvg{}'}{m^\bullet}}{w_\tl}
\end{array}
$$
\noindent Then,
$\mdl{\ell}{h'}{\st{\typesvg{}'}{m^\bullet+1}}{w_\hd::w_\tl}$ by the
definition.
\end{description}
\end{description}
\end{proof}

\begin{lemma}[The model relation for $v$  depends only on values in the footprint of $v$]\ \\
\label{LEMMA:ModelFootpr}For $v$, $h$, $w$, and $\typesvg$, the relation
$\mdl{v}{h}{\typesvg}{w}$ implies
$\mdl{v}{h|_{\Reg{h}{v}}}{\typesvg}{w}$.
\end{lemma}
\begin{proof} The lemma is proved by induction on $\typesvg$.
\begin{description}
\item[$\typesvg = \intt$] By the definition, $v$ is a constant $c$ and thus $w = c$.
Then $\mdl{v}{h|_{\Reg{h}{v}}}{\typesvg}{w}$.
\item[$\typesvg = \st{\typesvg{}}{n^\bullet}$]  We proceed by induction on $n^\bullet$.
\begin{description}
\item[$\typesvg = \st{\typesvg{}'}{0}$] By the definition $v = \nul$ and $w = \mbox{\texttt{[]}}$.
Then $\mdl{v}{h|_{\Reg{h}{v}}}{\typesvg}{w}$.
\item[$\typesvg = \st{\typesvg{}'}{m^\bullet + 1}$] By the definition $v=\ell$.
Then $\mdl{\ell}{h}{\st{\typesvg{}'}{m^\bullet+1}}{w}$ means
that $w=w_\hd::w_\tl$ for some $w_\hd$ and $w_\tl$, and
$$
\begin{array}{l}
\ell \in\dom{h},\\
\mdl{h.\ell.\hd}{\rstrloc{h}{\ell}}{\typesvg{}'}{w_\hd},\\
\mdl{h.\ell.\tl}{\rstrloc{h}{\ell}}{\st{\typesvg{}'}{m^\bullet}}{w_\tl}
\end{array}
$$
\noindent
We apply the induction assumption, with the heap $\rstrloc{h}{\ell}$:

$$
\begin{array}{l}
\ell \in\dom{h},\\
\mdl{h.\ell.\hd}{\rstrloc{h}{\ell}|_{\Reg{\rstrloc{h}{\ell}}{h.\ell.\hd}}}{\typesvg{}'}{w_\hd},\\
\mdl{h.\ell.\tl}{\rstrloc{h}{\ell}|_{\Reg{\rstrloc{h}{\ell}}{h.\ell.\tl}}}{\st{\typesvg{}'}{m^\bullet}}{w_\tl}
\end{array}
$$
\noindent

Due to $\Reg{\rstrloc{h}{\ell}}{h.\ell.\hd}\subseteq \dom{h}\setminus \{\ell\}$ (lemma \ref{LEMMA:FootprInHeap})
we have

$$
\begin{array}{l}
\restrict{\rstrloc{h}{\ell}}{{\Reg{\rstrloc{h}{\ell}}{h.\ell.\hd}}}=\\
=\restrict{h}{\Reg{\rstrloc{h}{\ell}}{h.\ell.\hd}}=\\
=\restrict{h}{\Reg{\rstrloc{h}{\ell}}{h.\ell.\hd}\setminus\{\ell\}}.\\
\end{array}$$

\noindent
Similarly $\restrict{\rstrloc{h}{\ell}}{{\Reg{\rstrloc{h}{\ell}}{h.\ell.\tl}}}=
\restrict{h}{\Reg{\rstrloc{h}{\ell}}{h.\ell.\tl}\setminus\{\ell\}}$.

Due to $\ell \in \Reg{h}{\ell}$, and lemma
\ref{LEMMA:ModelExtHeap} -- with $\Reg{\rstrloc{h}{\ell}}{h.\ell.\hd} \setminus\{\ell\}\subseteq
\Reg{h}{h.\ell.\hd}\setminus\{\ell\}$, we have

$$
\begin{array}{l}
\ell \in\dom{h_{\Reg{h}{\ell}}},\\
\mdl{\restrict{h}{\Reg{h}{\ell}}.\ell.\hd}{\restrict{h}{\Reg{h}{h.\ell.\hd}\setminus\{\ell\}}}{\typesvg{}'}{w_\hd},\\
\mdl{\restrict{h}{\Reg{h}{\ell}}.\ell.\tl}{\restrict{h}{\Reg{h}{h.\ell.\hd}\setminus\{\ell\}}}{\st{\typesvg{}'}{n^\bullet}}{w_\tl}
\end{array}
$$
\noindent Thus,
$\mdl{\ell}{h|_{\Reg{h}{\ell}}}{\st{\typesvg{}'}{m^\bullet+1}}{w_\hd::w_\tl}$.
\end{description}
\end{description}
\end{proof}

\begin{lemma}[Equality of footprints implies equivalence of model relations]\ \\
\label{LEMMA:EqFootprModels}
If $h|_{\Reg{h}{v}}=h'|_{\Reg{h}{v}}$ then
$\mdl{v}{h}{\typesvg}{w}$ implies $\mdl{v}{h'}{\typesvg}{w}$.
\end{lemma}
\begin{proof}
Assume  $\mdl{v}{h}{\typesvg}{w}$. Lemma \ref{LEMMA:ModelFootpr} states that
this implies $\mdl{v}{\restrict{h}{\Reg{h}{v}}}{\typesvg}{w}$.
Assuming $h|_{\Reg{h}{v}}=h'|_{\Reg{h}{v}}$ we get
$\mdl{v}{\restrict{h'}{\Reg{h}{v}}}{\typesvg}{w}$. Since
$\dom{\restrict{h'}{\Reg{h}{v}}}=\dom{\restrict{h}{\Reg{h}{v}}}=\Reg{h}{v}$
we have
$\restrict{h'}{\dom{\restrict{h'}{\Reg{h}{v}}}} =
\restrict{h'}{\Reg{h}{v}}$ and
we may apply lemma \ref{LEMMA:ModelExtHeap}, which
gives $\mdl{v}{h'}{\typesvg}{w}$.
\end{proof}

\begin{lemma}[Extending a store preserves the validity of the store]\ \\
\label{LEMMA:ChangeStore} Given a ground context $\Gamma^\bullet$, store $s$, heap $h$, value $v$, a set of variables
$\mathit{vars}$ and a variable $x \not\in \mathit{vars}$, s.t.
$x \not\in \dom{s}$,
one has
 $$\soundstore{\mathit{vars}}{\Gamma^\bullet}{s[x:=v]}{h} \Longleftrightarrow \soundstore{\mathit{vars}}{\Gamma^\bullet}{s}{h}$$
\end{lemma}
\begin{proof}
The lemma follows from the definition of $\soundstore{}{}{}{}$.
\end{proof}

\begin{lemma}[Weakening for valid stores]\ \\
\label{LEMMA:SubsetFV}Given a set of variables $\mathit{vars}_1$,
ground context $\Gamma^\bullet$, stack $s$, and heap $h$, for any set of variables $\mathit{vars}_2$ such that such that $\mathit{vars}_2 \subseteq \mathit{vars}_1$ one has
$$\soundstore{\mathit{vars}_1}{\Gamma^\bullet}{s}{h} \implies \soundstore{\mathit{vars}_2}{\Gamma^\bullet}{s}{h}$$
\end{lemma}
\begin{proof}
The lemma follows from the definition of $\soundstore{}{}{}{}$.
\end{proof}

\begin{lemma}[Validity for the disjoint union of sets of variables]
\label{LEMMA:UnionValid}
For any  store $s$ and a ground context $\Gamma^\bullet$  one has
$$\soundstore{\mathit{vars}_1 \cup \mathit{vars}_2}{\Gamma^\bullet}{s}{h}
\Longleftrightarrow \soundstore{\mathit{vars}_1}{\Gamma^\bullet}{s}{h} \;\land\;
\soundstore{\mathit{vars}_2}{\Gamma^\bullet}{s}{h}$$
\end{lemma}
\noindent
\begin{proof}
The lemma follows immediately from the definition of a valid store.
\end{proof} 
\vskip-30 pt

\end{document}